\documentclass[11pt,a4paper]{article}
\usepackage[margin=1in]{geometry}
\usepackage[english]{babel}
\usepackage[ruled,linesnumbered,vlined]{algorithm2e}
\usepackage{algpseudocode}
\usepackage{amsmath}
\usepackage{amssymb}
\usepackage{amsthm}
\usepackage{array}
\usepackage{booktabs}
\usepackage{framed}
\usepackage{graphicx}
\usepackage{float}
\usepackage{makecell}
\usepackage{mathtools}
\usepackage{multirow}
\usepackage{subcaption}
\usepackage{tcolorbox}
\usepackage[table]{xcolor}
\usepackage{tikz}
\usetikzlibrary{arrows.meta,positioning}
\usepackage[colorlinks=true, allcolors=blue]{hyperref}
\usepackage{cleveref}
\usetikzlibrary{arrows.meta,patterns,positioning}

\usepackage{pgfplots}
\usepackage{pgfplotstable}
\pgfplotsset{
    compat=1.17,
    name nodes near coords/.style={
        every node near coord/.append style={
            name=#1-\coordindex,
            alias=#1-last,
        },
    },
    name nodes near coords/.default=coordnode,
}

\definecolor{shadecolor}{gray}{0.9}
\DeclarePairedDelimiter\floor{\lfloor}{\rfloor}

\newtheorem{theorem}{Theorem}[section]
\newtheorem{definition}[theorem]{Definition}

\newtheorem{lemma}[theorem]{Lemma}

\newtheorem{example}[theorem]{Example}

\newcommand{\E}{{\mathbb{E}}}

\newcommand{\bX}{\mathbf{X}}

\newcommand{\bc}{\mathbf{c}}

\newcommand{\costbound}[1]{\Gamma_{#1}}
\newcommand{\excess}[1]{E_{#1}}

\newenvironment{proofof}[1]{{\vspace*{5pt} \noindent\bf Proof of #1:  }}{\hfill\rule{2mm}{2mm}\vspace*{5pt}}

\title{Truthful-in-Expectation MMS Allocations for Chores}
\author{Zehan Lin \and Biaoshuai Tao \and Xiaowei Wu \and Yuhao Zhang}
\date{\today}

\begin{document}

\maketitle

\begin{abstract}
We study truthful-in-expectation (TIE) mechanisms for allocating indivisible chores alongside ex-post maximin share (MMS) guarantees.
For goods, Bu and Tao (FOCS 2025) established a $(1/n)$-approximation for TIE mechanisms, and this was substantially improved by Babaioff, Feige, and Manaker Morag (FOCS 2026), who established an $\Omega(1/\log n)$ approximation, where $n$ is the number of agents. 
The corresponding problem for chores has received less attention.
The best-known result is due to Aziz, Li, and Wu (MAPR 2024), who gave a TIE mechanism with an $O(\sqrt{\log n})$ MMS approximation guarantee that holds only in expectation.
They also established a $6/5$ lower bound for TIE mechanisms for two agents. 
In this paper, we present the first TIE mechanism for chores that achieves a constant ex-post MMS approximation guarantee.
Specifically, our mechanism guarantees an ex-post ratio of $1.97$ for any number of agents $n$, which improves to $4/3$ for $n=2$ and $3/2$ for $n=3$.
On the hardness side, we tighten the two-agent lower bound to $4/3$, showing that our mechanism is optimal for $n=2$. More generally, we establish a lower bound of $13/12$ on the ex-post MMS approximation ratio achievable by TIE mechanisms for every $n\ge3$.
\end{abstract}


\newpage

\section{Introduction}
Fair division, originating from the seminal work of Steinhaus~\cite{journals/econometrica/Steinhaus48}, is a fundamental problem in economics and computer science.
The central goal of the problem is to fairly allocate a set $M$ of $m$ items among a set $N$ of $n$ agents with heterogeneous preferences.
These items may be goods, which yield non-negative utilities (e.g., resources), or chores, which impose costs (e.g., tasks).
In this paper, we study the allocation of indivisible chores with additive cost functions.


Fairness notions are broadly divided into two main categories: envy-based notions, which evaluate an agent's bundle relative to those of others, and share-based notions, which evaluate an agent's bundle against a benchmark determined by her own valuation and the number of agents.
Among envy-based notions, envy-freeness (EF)~\cite{books/yale/Foley66} is the classic standard, which requires every agent to weakly prefer her own bundle to that of any other agent. 
However, EF can easily fail in the presence of indivisible chores. 
For instance, a single chore assigned to one of two agents immediately induces envy.
This limitation has led to the widely studied relaxation of envy-freeness up to one item (EF1)~\cite{conf/sigecom/LiptonMMS04}, which guarantees that any envy can be eliminated after removing at most one chore from the envious agent's bundle.
Regarding share-based notions, a fundamental requirement is proportionality (PROP), which requires that each agent receives at most $1/n$ of her total cost for all chores. 
Since indivisibility similarly makes PROP impossible to guarantee, Budish~\cite{journals/jpe/Budish11} introduced the maximin share (MMS). 
For chores, an agent's MMS is defined as the cost she can secure by partitioning the chores into $n$ bundles and receiving her least preferred one.
MMS allocations need not exist even under additive functions~\cite{journals/jacm/KurokawaPW18, conf/aaai/AzizRSW17,conf/wine/FeigeST21,misc/arxiv/EzraG26}.
Hence, a substantial body of work focuses on approximate MMS allocations, which bound how much an agent's cost may exceed her MMS guarantee.

In addition to fairness, allocation mechanisms must account for agents' incentives.
In many real-world scenarios, the preferences of agents are typically private information, and an agent may misreport them to obtain a more favorable bundle.
A mechanism is truthful if no agent can benefit from misreporting her preferences, regardless of the reports of the other agents.
Achieving both truthfulness and fairness is challenging.
For two agents with additive valuations over goods, deterministic truthful mechanisms cannot guarantee more than an $O(1/m)$ fraction of each agent's MMS~\cite{conf/sigecom/AmanatidisBCM17}.
For chores, the corresponding characterization yielded a tight deterministic approximation ratio of $2-1/\floor{m/2}$ for two agents and $m\ge2$~\cite{conf/sagt/LiTWWYZ25}.
As this ratio approaches $2$ with increasing $m$, no deterministic truthful mechanism can guarantee a ratio bounded away from $2$ even for two agents.

Randomization offers a way to obtain stronger fairness guarantees while preserving incentives in expectation.
A mechanism is truthful in expectation (TIE) if no agent can improve her expected utility, or reduce her expected cost, by misreporting.
In this work, we therefore seek to combine this incentive requirement with ex-post fairness, so that the fairness guarantee holds for every realized allocation.
For the setting of goods, Babaioff, Ezra, and Feige~\cite{conf/wine/BabaioffEF22} explicitly asked whether a TIE mechanism could guarantee each agent a constant fraction of her MMS ex-post.
Subsequent work obtained TIE mechanisms with ex-post MMS guarantees that depend on the number of agents.
Specifically, Bu and Tao~\cite{conf/focs/BuT25} designed a TIE mechanism that guaranteed each agent at least a $1/n$ fraction of her MMS ex-post.
More recently, Babaioff, Feige, and Manaker Morag~\cite{journals/corr/abs-2604-27211} gave an ordinal TIE mechanism that guaranteed each agent an $\Omega(1/\log n)$ fraction of her MMS ex-post, and obtained a $2/3$-MMS guarantee for two agents.

For chores, Aziz, Li, and Wu~\cite{journals/mp/AzizLW24} studied approximate MMS allocation using ordinal preferences, including the case of strategic agents.
They provided a TIE mechanism whose approximation ratio is $O(\sqrt{\log n})$ in expectation. 
At a high level, their randomized mechanism first distributed chores independently and uniformly at random, and then allowed agents to return chores of high cost for redistribution.
Nevertheless, their approximation guarantee is not constant, and it holds only in expectation, rather than for every realized allocation.
This directly leads to the following open question:
\begin{center}\begin{minipage}{0.98\linewidth}
\noindent\textbf{Main Question.}
\textit{Can we design a TIE mechanism for indivisible chores that guarantees a constant MMS approximation ratio for every realized allocation?}
\end{minipage}\end{center}

\subsection{Our Results}
\label{sec:introduction}
In this paper, we answer the question affirmatively by designing TIE mechanisms with constant ex-post MMS approximation guarantees.
Furthermore, our mechanisms also ensure ex-ante EF, which guarantees each agent's expected cost for her own bundle is no greater than her expected cost for any other agent's bundle.

We first present a simple mechanism with a $(2-1/n)$-MMS guarantee and then show how to improve it via the technique of \emph{nominations}. 
Specifically, the mechanism first asks each agent to report her $k$ largest chores. 
After receiving these nominations, $n-k$ randomly selected agents receive only their non-nominated chores, whereas each of the remaining $k$ agents gets one leftover chore.
For $k < n/2$, this nomination mechanism achieves an ex-post $(2-1/(n-k))$-MMS guarantee (\Cref{thm:top-avoidance}).
For two agents, a separate construction based on nominations achieves $4/3$-MMS (\Cref{thm:two-agent-nomination}), which matches our lower bound and hence achieves the optimal approximation ratio.
These guarantees improve on the $O(\sqrt{\log n})$ approximation in expectation obtained by Aziz, Li, and Wu~\cite{journals/mp/AzizLW24} and provide a constant bound for every realized allocation.

However, the nomination guarantee still approaches $2$ as $n$ grows, and hard instances suggest that nomination mechanisms cannot overcome this limitation. 
To achieve an approximation ratio strictly better than $2$, we design a randomization over allocation rules that balance large and small chores differently. 
Intuitively, an agent with a higher likelihood of being assigned large chores is compensated with fewer small chores, whereas an agent who avoids large chores absorbs a larger share of small ones. 
We coordinate these complementary rules to ensure that every agent is assigned each chore with marginal probability $1/n$ overall, which ensures both TIE and ex-ante EF.
We represent the probability distributions under each rule as fractional allocations, which are subsequently rounded into integral assignments while controlling every agent's realized cost.
This construction guarantees ex-post $1.97$-MMS for every $n\ge19$.
Combining it with the nomination mechanisms for smaller $n$ yields our main positive result.\footnote{A more involved refinement improves the bound to $15/8=1.875$. See \Cref{sec:conclusion} for a more detailed discussion.}

\begin{tcolorbox}[colback=gray!15, colframe=gray!15]
\textbf{Result 1} (Theorems~\ref{thm:top-avoidance} and~\ref{thm:1.97}).\label{Result1}
\emph{For every $n\ge2$, there exists a TIE mechanism for indivisible chores with additive costs that is ex-ante EF and ex-post $1.97$-MMS.}
\end{tcolorbox}


Our result not only shows that better-than-$2$ ex-post MMS approximation is possible for TIE mechanisms, but also demonstrates the power of randomization, as the achieved ratio strictly beats the impossibility bound of $2$ for deterministic truthful mechanisms.

To complement these positive results, we establish lower bounds for TIE mechanisms bounded away from $1$ for every $n$.

\begin{tcolorbox}[colback=gray!15, colframe=gray!15]
\textbf{Result 2} (\Cref{thm:tie-lower-bounds}).\label{Result2}
\emph{For every $n\ge2$, no TIE mechanism for indivisible chores with additive costs can guarantee an ex-post MMS approximation ratio strictly smaller than $13/12$.
For $n=2,3$, the lower bound improves to $4/3$ and $8/7$, respectively.}
\end{tcolorbox}

\Cref{fig:tie-ratios} summarizes these upper and lower bounds through $n=80$, including the switch between the two upper-bound constructions at $n=66$.
The bounds coincide for two agents, while closing the gap for $n\ge3$ remains open.

\begin{figure}[!htbp]
\centering
\begin{tikzpicture}
\begin{axis}[
  width=0.75\linewidth, height=6.5cm,
  xmin=1, xmax=81, ymin=1.03, ymax=2.10,
  xlabel={Number of agents $n$},
  ylabel={Ex-post MMS ratio},
  xtick={2,20,40,60,66,80},
  ytick={1.1,1.3,1.5,1.7,1.9},
  tick label style={font=\small},
  label style={font=\small},
  grid=major, grid style={gray!18},
  axis line style={gray!60},
  legend style={at={(0.97,0.48)},anchor=east,draw=none,fill=white,font=\small},
  legend cell align=left
]
\addplot[blue!70!black,thick,mark=*,mark size=1.2pt] coordinates {
(2,1.3333333333) (3,1.5000000000) (4,1.6666666667) (5,1.6666666667) (6,1.7500000000) (7,1.7500000000) (8,1.8000000000)
(9,1.8000000000) (10,1.8333333333) (11,1.8333333333) (12,1.8571428571) (13,1.8571428571) (14,1.8750000000) (15,1.8750000000)
(16,1.8888888889) (17,1.8888888889) (18,1.9000000000) (19,1.9000000000) (20,1.9090909091) (21,1.9090909091) (22,1.9166666667)
(23,1.9166666667) (24,1.9230769231) (25,1.9230769231) (26,1.9285714286) (27,1.9285714286) (28,1.9333333333) (29,1.9333333333)
(30,1.9375000000) (31,1.9375000000) (32,1.9411764706) (33,1.9411764706) (34,1.9444444444) (35,1.9444444444) (36,1.9473684211)
(37,1.9473684211) (38,1.9500000000) (39,1.9500000000) (40,1.9523809524) (41,1.9523809524) (42,1.9545454545) (43,1.9545454545)
(44,1.9565217391) (45,1.9565217391) (46,1.9583333333) (47,1.9583333333) (48,1.9600000000) (49,1.9600000000) (50,1.9615384615)
(51,1.9615384615) (52,1.9629629630) (53,1.9629629630) (54,1.9642857143) (55,1.9642857143) (56,1.9655172414) (57,1.9655172414)
(58,1.9666666667) (59,1.9666666667) (60,1.9677419355) (61,1.9677419355) (62,1.9687500000) (63,1.9687500000) (64,1.9696969697)
(65,1.9696969697) (66,1.9700000000) (67,1.9700000000) (68,1.9700000000) (69,1.9700000000) (70,1.9700000000) (71,1.9700000000)
(72,1.9700000000) (73,1.9700000000) (74,1.9700000000) (75,1.9700000000) (76,1.9700000000) (77,1.9700000000) (78,1.9700000000)
(79,1.9700000000) (80,1.9700000000)
};
\addlegendentry{Upper bound}
\addplot[orange!85!black,thick,dashed,mark=square*,mark size=1.1pt] coordinates {
(2,1.3333333333) (3,1.1428571429) (4,1.0909090909) (5,1.0833333333) (6,1.0833333333) (7,1.0833333333) (8,1.0833333333)
(9,1.0833333333) (10,1.0833333333) (11,1.0833333333) (12,1.0833333333) (13,1.0833333333) (14,1.0833333333) (15,1.0833333333)
(16,1.0833333333) (17,1.0833333333) (18,1.0833333333) (19,1.0833333333) (20,1.0833333333) (21,1.0833333333) (22,1.0833333333)
(23,1.0833333333) (24,1.0833333333) (25,1.0833333333) (26,1.0833333333) (27,1.0833333333) (28,1.0833333333) (29,1.0833333333)
(30,1.0833333333) (31,1.0833333333) (32,1.0833333333) (33,1.0833333333) (34,1.0833333333) (35,1.0833333333) (36,1.0833333333)
(37,1.0833333333) (38,1.0833333333) (39,1.0833333333) (40,1.0833333333) (41,1.0833333333) (42,1.0833333333) (43,1.0833333333)
(44,1.0833333333) (45,1.0833333333) (46,1.0833333333) (47,1.0833333333) (48,1.0833333333) (49,1.0833333333) (50,1.0833333333)
(51,1.0833333333) (52,1.0833333333) (53,1.0833333333) (54,1.0833333333) (55,1.0833333333) (56,1.0833333333) (57,1.0833333333)
(58,1.0833333333) (59,1.0833333333) (60,1.0833333333) (61,1.0833333333) (62,1.0833333333) (63,1.0833333333) (64,1.0833333333)
(65,1.0833333333) (66,1.0833333333) (67,1.0833333333) (68,1.0833333333) (69,1.0833333333) (70,1.0833333333) (71,1.0833333333)
(72,1.0833333333) (73,1.0833333333) (74,1.0833333333) (75,1.0833333333) (76,1.0833333333) (77,1.0833333333) (78,1.0833333333)
(79,1.0833333333) (80,1.0833333333)
};
\addlegendentry{Lower bound}
\draw[gray!65,densely dashed] (axis cs:66,1.03) -- (axis cs:66,1.97);
\node[anchor=north east,font=\small,text=blue!70!black] at (axis cs:80,2.08) {$1.97$};
\node[anchor=north east,font=\small,align=right] (switch) at (axis cs:64,1.86) {Switch at $n=66$};
\draw[->,gray!75] (switch.north east) -- (axis cs:66,1.97);
\end{axis}
\end{tikzpicture}
\caption{Upper and lower bounds on the ex-post MMS ratio by TIE mechanisms. }
\label{fig:tie-ratios}
\end{figure}
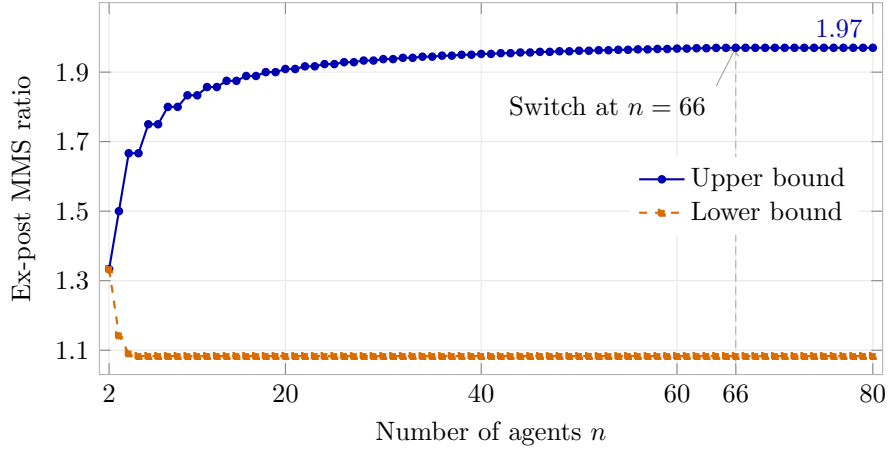

\subsection{A Technical Overview}

We provide an overview of the main techniques used in this work.

\paragraph{Fixing Marginals Before Designing the Allocation.}
Aziz, Li, and Wu~\cite{journals/mp/AzizLW24} observe that assigning every chore independently and uniformly at random gives every agent probability $1/n$ of receiving every chore.
These equal marginals make an agent's expected cost independent of her report and are therefore useful for truthfulness.
However, because the marginals arise from the random assignment, different chores remain substantially independent, and an unlucky realization may assign several costly chores to the same agent.
This also explains why their approximation ratio holds only in expectation.
We reverse this perspective.
Rather than specifying a random allocation first and deriving its marginals, we first prescribe the marginal probability of every agent--chore pair, equivalently a fractional allocation, and then decompose it into integral allocations.
This lets us coordinate the assignments of different chores so that \emph{every} allocation in the support satisfies a useful cost bound while preserving the desired incentive guarantee.
For example, under the uniform fractional allocation $x_{ie}=1/n$, sorting chores by cost and applying the Birkhoff--von Neumann decomposition ensures that each agent receives at most one chore from each consecutive block of $n$ chores, yielding the $(2-1/n)$-MMS mechanism.
The $k$-nomination mechanism uses the same principle with different marginals: nominating a chore reduces its marginal probability by the same amount, so nominating the $k$ largest chores minimizes expected cost.
We then decompose these marginals so that an agent either receives a singleton chore or avoids all nominations and receives at most one chore from each remaining block.
This yields the ex-post $\left(2-1/(n-k)\right)$-MMS guarantee.
Thus, the decomposition turns carefully chosen truthful marginals into a lottery whose entire support is well behaved.

\paragraph{Trading Fractional Cost Against Rounding Error.}
To obtain a ratio bounded away from $2$, we exploit more structure in the decomposition.
The sorted-slot analysis shows that the realized cost of an agent is controlled by two quantities.
The first is her \emph{fractional cost}.
The second is the \emph{rounding error}: roughly speaking, one additional costly chore may be contributed by the first slot, and this term is bounded by the largest chore having positive fraction.
Under the uniform fractional allocation, both quantities are at most the agent's MMS.
Consequently, the approximation ratio can be close to $2$ only when both are simultaneously close to the MMS.
We show that this can happen only for agents with a small number of very costly chores.
We call these agents \emph{hard}.
Our key observation is that, for a hard agent, moving fractional mass between these few very costly chores and the remaining chores helps in \emph{either direction}.
Consider, for example, two hard agents.
In one intermediate fractional allocation, the first agent receives zero fraction of her large chores and a larger fraction of the remaining chores.
Her fractional cost increases, but her rounding error drops sharply because none of her large chores is eligible after rounding.
The second agent takes the complementary role: she receives larger fractions of the large chores of the pair and smaller fractions of the remaining chores.
Her rounding error may remain as large as her MMS, but her fractional cost decreases enough to compensate for it.
By randomizing these two roles, each agent still receives every chore with the desired average marginal probability, while in every intermediate allocation the sum of fractional cost and rounding error is strictly smaller than the original bound.
By applying the Birkhoff--von Neumann decomposition of the intermediate fractional allocation, we can simultaneously preserve the designed marginals and convert the fractional trade-off into an ex-post guarantee for every realized allocation, yielding the $1.97$-MMS mechanism.

\subsection{Related Work}
Given the extensive literature on fair division, we focus on work most closely related to (approximate) MMS allocations and truthful mechanisms.
For a broader overview, we refer the reader to the surveys~\cite{journals/sigecom/AzizLMW22,journals/ai/AmanatidisABFLMVW23,guo2023survey}.

\paragraph{Truthful Mechanisms for Indivisible Items.}
For chores, beyond the randomized guarantee discussed above, Aziz, Li, and Wu~\cite{journals/mp/AzizLW24} obtained a deterministic truthful $O(\log(m/n))$-MMS mechanism using ordinal preferences.
Sun and Chen~\cite{journals/eor/SunC25} designed a TIE mechanism with ex-ante EF, ex-post EF1, and Pareto optimality (PO) under restricted additive costs, where each chore incurs an inherent cost or zero.
For goods, beyond the MMS guarantees discussed above, Bu and Tao~\cite{conf/focs/BuT25} obtained TIE mechanisms guaranteeing ex-post EF1 for two agents and envy-freeness up to the removal of $O(\sqrt{n})$ goods for general $n$; their $1/n$-MMS mechanism also satisfies ex-post PROP1.
Babaioff, Feige, and Manaker Morag~\cite{journals/corr/abs-2604-27211} further obtained an ex-post $\Omega(1/\log\log n)$-MMS guarantee using cardinal information and approximate TIE.
A line of work characterized deterministic truthful mechanisms for allocating indivisible goods under various preference domains and additional requirements, yielding limitations on attainable fairness guarantees~\cite{conf/sigecom/AmanatidisBCM17,conf/atal/HosseiniL19,conf/sigecom/BabaioffM25}.
Under binary additive valuations, the maximum Nash welfare rule with lexicographic tie-breaking is group strategyproof, EF1, and PO~\cite{conf/wine/0002PP020}; Barman and Verma~\cite{conf/aaai/BarmanV22} extended these guarantees to matroid-rank valuations.
For submodular valuations with binary marginal values, Babaioff, Ezra, and Feige~\cite{conf/aaai/BabaioffEF21} obtained a deterministic truthful Lorenz-dominating allocation rule and showed that randomizing priorities additionally ensures ex-ante EF.
Truthful mechanism design has also been studied for house allocation with initial endowments or distributional constraints~\cite{journals/scw/Svensson99,journals/jet/ShendeP23,journals/ai/SuzukiTYYZ23}.

\paragraph{Truthful Mechanisms for Divisible Items.}
Truthful mechanisms have also been studied for divisible goods~\cite{journals/mansci/FreemanWVP24}.
A related model is cake cutting, where a heterogeneous divisible resource is represented by an interval and agents may value its parts differently~\cite{journals/geb/ChenLPP13,conf/wine/AzizY14}.
In this setting, Bu, Song, and Tao~\cite{journals/ai/BuST23} proved that no deterministic mechanism could be both truthful and PROP, even for two agents with piecewise-constant valuations.
Positive results were obtained under more restricted preferences: \cite{journals/geb/ChenLPP13} designed a deterministic truthful and EF mechanism for piecewise-uniform valuations, allowing free disposal.
Randomization also enabled stronger guarantees, as Mossel and Tamuz~\cite{conf/sagt/MosselT10} obtained a TIE mechanism guaranteeing PROP.


\paragraph{(Approximate) MMS Allocations for Chores.}
Under the cardinal model, an MMS allocation for two agents can be obtained via divide-and-choose. 
However, MMS allocations are not guaranteed to exist for three or more agents~\cite{conf/aaai/AzizRSW17,conf/wine/FeigeST21,misc/arxiv/EzraG26}.
The current state-of-the-art lower bound on the approximation ratio is $31/30$, established by Ezra and Garbuz~\cite{misc/arxiv/EzraG26} using a four-agent instance.
These negative results motivated the study of approximate MMS allocations~\cite{conf/aaai/AzizRSW17,journals/mp/AzizLW24,conf/sigecom/FeigeH23,journals/teco/BarmanK20,conf/sigecom/HuangL21}.
The current best approximation ratio for chores is $13/11$~\cite{conf/sigecom/HuangS23}.
When only ordinal information is available, the attainable guarantees are inherently weaker, with the currently known ratios being $2-1/n$~\cite{conf/aaai/AzizRSW17}, $5/3$~\cite{journals/mp/AzizLW24}, and $8/5$~\cite{conf/sigecom/FeigeH23}.

\section{Preliminaries}
\label{sec:preliminary}

We study the problem of fairly allocating a set $M$ of $m$ indivisible chores among the agents $N=[n]$.
We refer to any subset $X\subseteq M$ as a bundle.
Each agent $i\in N$ has a nonnegative additive cost function $c_i:2^M\to\mathbb{R}_{\ge0}$.
We use $\bc=\left(c_1,\ldots,c_n\right)$ for a reported cost profile.
For $X\subseteq M$ and a chore $e\in M$, write $X+e=X\cup\left\{e\right\}$ and $X-e=X\setminus\left\{e\right\}$.
For $M'\subseteq M$, a partition is a collection of pairwise disjoint subsets whose union is $M'$.
An allocation $\bX=\left(X_1,\ldots,X_n\right)$ is an ordered partition of $M$, where agent $i$ receives $X_i$.
For every positive integer $t$, write $[t]=\left\{1,2,\ldots,t\right\}$.
For every real number $x$, define $\left(x\right)^+:=\max\left\{0,x\right\}$.
We interpret the maximum over an empty set as zero.
We fix a public tie-breaking order on the chores once and for all, e.g., by chore index.
For every agent $i$ and $t\in[m]$, let $H_{i,[t]}$ denote the set of the $t$ most costly chores of agent $i$, and let $c_{i,(t)}$ denote the cost of her $t$-th most costly chore under this tie-breaking rule.
Thus, $H_{i,[t]}\subseteq H_{i,[t+1]}$ and $c_{i,(1)}\ge\cdots\ge c_{i,(m)}$.

\begin{definition}[MMS]
Given $M'\subseteq M$, the maximin share of agent $i$ for $M'$ is
\begin{equation*}
\mu_i\left(M'\right):= \min_{\bX\in\Pi_n\left(M'\right)} \max_{j\in N}\left\{c_i\!\left(X_j\right)\right\},
\end{equation*}
where $\Pi_n\left(M'\right)$ denotes the set of all ordered $n$-partitions of $M'$, with empty bundles allowed.
When $M$ is clear from context, write $\mu_i:=\mu_i\left(M\right)$.
\end{definition}

In the strategic setting, the reported cost function of agent $i$ might be different from $c_i$. 
We use $\widehat c_i$ for an alternative report when it is useful to distinguish the two.
We write $X_i\left(\bc\right)$ for agent $i$'s random bundle at profile $\bc$.
All mechanisms allocate every chore and use no payments.

\begin{definition}[Truthfulness in expectation]\label{def:tie}
A randomized mechanism is truthful in expectation (TIE) if, for every agent $i$, every true additive cost $c_i$, every alternative report $\widehat c_i$, and every fixed profile $\bc_{-i}$ of the other reports,
\begin{equation*}
\E\left[c_i\!\left(X_i\left(c_i,\bc_{-i}\right)\right)\right] \le \E\left[c_i\!\left(X_i\left(\widehat c_i,\bc_{-i}\right)\right)\right].
\end{equation*}
In other words, a mechanism is TIE if truthful reporting gives the lowest expected cost. 
\end{definition}

\begin{definition}[Ex-post MMS approximation]\label{def:ex-post-mms}
For $\alpha\ge1$, a mechanism is ex-post $\alpha$-MMS if, at every report profile $\bc$, each allocation $\bX$ in its support satisfies
\begin{equation*}
c_i\!\left(X_i\right)\le\alpha\cdot\mu_i \qquad\text{for every agent }i\in N.
\end{equation*}
\end{definition}

We remark that all MMS approximation ratios introduced in this paper are ex-post ratios unless stated otherwise.
A fractional allocation is a nonnegative matrix $x=\left(x_{ie}\right)_{i\in N,e\in M}$ that satisfies $\sum_{i\in N}x_{ie}=1$ for every chore $e$.
The \emph{uniform fractional allocation} is the fractional allocation with $x_{ie}=1/n$ for every agent $i \in N$ and chore $e \in M$.
A randomized mechanism is \emph{equal-marginal} if, at every report profile, it holds that
\begin{equation*}
\Pr\left[e\in X_i\right]=\frac1n \qquad\text{for every agent }i\in N\text{ and chore }e\in M.
\end{equation*}
Thus, an equal-marginal mechanism implements the uniform fractional allocation at the level of unconditional marginals.
A key property of equal-marginal is that it ensures ex-ante EF: for every pair of agents $i,j\in N$, we have $\E[c_i(X_j)] = \frac{1}{n}\cdot c_i(M) = \E[c_i(X_i)]$.

\begin{lemma}[Basic MMS bounds]\label{lem:basic-mms}
For every agent $i \in N$, we have
\begin{equation*}
c_i\!\left(M\right)\le n\cdot\mu_i, \qquad c_i\!\left(e\right)\le\mu_i \quad\text{for every }e\in M.
\end{equation*}
\end{lemma}

Note that if $\mu_i=0$, then \Cref{lem:basic-mms} implies $c_i\!\left(M\right)=0$, in which case we can assign all chores to agent $i$ (which preserves TIE). 
Therefore, throughout this paper, we assume $\mu_i > 0$ for every agent $i\in N$.

\section{\texorpdfstring{$k$}{k}-Nomination Mechanism}
\label{sec:nomination_overview}
\label{sec:nomination}

In this section, we begin by introducing a simple mechanism and establish its $(2-1/n)$-MMS and TIE guarantees (Lemma~\ref{lem:zero-nomination}). 
To achieve a better approximation ratio, we then propose the $k$-nomination mechanism that satisfies a stronger $2-1/(n-k)$-MMS guarantee (Theorem~\ref{thm:top-avoidance}), alongside showing the tightness of this bound (Lemma~\ref{prop:nomination-tight}).
In the simple $(2-1/n)$-MMS mechanism, the mechanism has the \emph{equal marginal} property: each agent receives each chore with probability $1/n$.
The stronger $2-1/(n-k)$-MMS mechanism uses a \emph{nomination} technique that adjusts the marginal probabilities for each agent's $k$ largest chores while guaranteeing TIE.


\subsection{The 0-Nomination Mechanism}
\label{sec:2-1/n}
We first introduce the 0-nomination mechanism, which is TIE and guarantees $(2-1/n)$-MMS ex-post.
As mentioned, the mechanism has the equal marginal property.
This property straightforwardly implies TIE, and is a natural idea for designing TIE mechanisms.
Our result indicates that constant bounds on MMS are already possible by this natural idea.
This is in sharp contrast to the setting with goods, as Babaioff, Feige, and Manaker Morag~\cite{journals/corr/abs-2604-27211} show that mechanisms with this property can achieve at most $1/n$-MMS.

Before collecting reports, we add dummy chores, whose costs are fixed at zero for every agent, so that $m$ is divisible by $n$.
Adding these chores does not change any agent's MMS, and we will remove them from the final allocation.
For each agent $i \in N$, we order the chores from highest to lowest reported cost, breaking ties by a fixed index order.
Assume that $m = b\cdot n$ for some integer $b\geq 1$.
We partition this sequence of chores into consecutive blocks $B_{i,1}, B_{i,2}, \ldots,B_{i,b}$, each of which has size $n$.

Mechanism~\ref{alg:zero-nomination} constructs a bipartite graph $G=(L\cup R,E)$, where $L=M$ and $R$ consists of all agents' blocks.
We remark that blocks belonging to different agents are regarded as distinct vertices, i.e., $|L| = |R| = m$.
We add an edge between each chore and the block containing it for each agent.
Therefore, every chore vertex has degree $n$, and every block vertex has degree $n$, so $G$ is an $n$-regular bipartite graph.
By Hall's theorem~\cite{journals/jlms/Hall35}, $G$ has a perfect matching.
Removing it leaves an $(n-1)$-regular bipartite graph. Repeating this argument decomposes $E$ into $n$ perfect matchings.
The mechanism then samples one of these matchings uniformly at random to determine the allocation.

\SetAlgorithmName{Mechanism}{Mechanism}{List of Mechanisms}
\begin{algorithm}[htp]
    \caption{0-Nomination Mechanism}
    \label{alg:zero-nomination}
    \KwIn{Reported cost functions $\bc$}
    Add public zero-cost dummy chores to $M$ so that $|M|$ is divisible by $n$\;
    $m\gets |M|$, $b\gets m/n$\;
    For each $i\in N$, sort chores by non-increasing reported cost and form blocks $B_{i,1},\ldots,B_{i,b}$ of size $n$\;
    $L\gets M$, $R\gets\{B_{i,r}:i\in N,\ r\in[b]\}$\;
    $E\gets\{(e,B_{i,r}):i\in N,\ r\in[b],\ e\in B_{i,r}\}$\;
    Decompose $E=E_1\cup\cdots\cup E_n$ into perfect matchings\;
    Sample a matching $E'$ uniformly at random from $\{E_1, \dots, E_n\}$\;
    \For{each $i \in N$} {
    $X_i \gets \{e \in M : (e, B_{i,r}) \in E' \text{ for some } r \in [b]\}$\;
    }
    Remove all dummy chores from $\bX$\;
    \KwOut{An allocation $\bX$}
\end{algorithm}

\medskip
We now establish the guarantees of Mechanism~\ref{alg:zero-nomination}.
\begin{lemma}\label{lem:zero-nomination}
    Mechanism~\ref{alg:zero-nomination} is TIE, ex-ante EF, and
    ex-post $(2-1/n)$-MMS.
\end{lemma}
\begin{proof}
    Fix an arbitrary reported profile and consider the graph $G = (L \cup R, E)$ constructed
    by Mechanism~\ref{alg:zero-nomination}. 
    As shown above, $G$ is $n$-regular, and its edges decompose into $n$ perfect matchings.
    Consequently, every sampled matching assigns each chore exactly once and gives each agent exactly one chore from each of her blocks.

    To prove truthfulness, observe that for any agent $i$ and any chore $e$, the edge connecting $e$ to the corresponding block of agent $i$ belongs to exactly one of the $n$ matchings.
    Consequently, under uniform sampling, we have $\Pr[e \in X_i] = 1/n$, which is independent of the reported profile. 
    In particular, regardless of the other agents' reports, for every $j\in N$, we always have
    \begin{equation*}
    \E[c_i(X_j)] = \sum_{e \in M} c_i(e) \cdot \Pr[e \in X_j] = \frac{c_i(M)}{n}.
    \end{equation*}
    Since agent $i$'s report cannot change this expected cost, the mechanism is TIE and ex-ante EF.

    It remains to bound the cost of each agent's realized bundle.
    Fix an arbitrary agent $i \in N$.
    Since agent $i$ receives at most one chore from each block $B_{i,r}$, we have $c_i\!\left(X_i\right)  \le \sum_{r=1}^{b} c_{i,\left((r-1)n+1\right)}$.
    Furthermore, since chores within each block are ordered by non-increasing cost, for all $r<b$, the cost of each chore in block $B_{i,r}$ is at least $c_{i,\left(rn+1\right)}$. It follows that
    \begin{equation*}
c_i(M)=\sum_{r\in[b]} \sum_{e\in B_{i, r}} c_i(e) \ge c_{i,\left(1\right)} + n\cdot \sum_{r=2}^b c_{i,\left((r-1)n+1\right)},
\end{equation*}
    Since $c_{i,\left(1\right)} \le \mu_i$ and $c_i(M) \le n\cdot \mu_i$ (by Lemma~\ref{lem:basic-mms}), rearranging the previous inequality yields $\sum_{r=2}^b c_{i,\left((r-1)n+1\right)} \leq (c_i(M) - c_{i,\left(1\right)})/{n}$. 
    Consequently, we have
    \begin{equation*}
c_i(X_i) \le \sum_{r=1}^b c_{i,\left((r-1)n+1\right)} \le c_{i,\left(1\right)} + \frac{c_i(M) - c_{i,\left(1\right)}}{n} \le \left(2 - \frac{1}{n}\right)\cdot \mu_i.
\end{equation*}
    This completes the proof of this lemma.
\end{proof}

\subsection{The \texorpdfstring{$k$}{k}-Nomination Mechanism}
\label{sec:nomination-many}
Next, we turn our attention to designing mechanisms with better approximation guarantees.
Intuitively, Mechanism~\ref{alg:zero-nomination} gives each agent at most one chore from each sorted block, but an agent who receives her largest chore may still receive a chore from every other block.
To see how this can lead to a ratio of $2-1/n$, consider some agent $i \in N$ and her cost function in Table~\ref{tab:zero-nomination-example}. 
There are $m=n\cdot\left(n-1\right)+1$ chores, with $c_i(e_1)=1$ and $c_i(e)=1/n$ for $e \in M\setminus \{e_1\}$.

\begin{table}[ht]
    \centering
    \renewcommand{\arraystretch}{1.2}
    \setlength{\tabcolsep}{12pt}
    \begin{tabular}{@{}lcccc@{}}
        \toprule
        & $e_1$ & $e_2$ & $\dots$ & $e_m$ \\
        \midrule
        Agent $i$ & $1$ & $1/n$ & $\dots$ & $1/n$ \\
        \bottomrule
    \end{tabular}
    \caption{A hard instance for Mechanism~\ref{alg:zero-nomination}.}
    \label{tab:zero-nomination-example}
\end{table}

In this instance, we have $\mu_i = 1$. However, by Mechanism~\ref{alg:zero-nomination}, under the worst-case realization of the allocation, agent $i$ may receive the largest chore $e_1$ along with one chore from each of the remaining $n-1$ blocks. 
This yields a total cost of $c_i(X_i)=1+(n-1)/n=(2-1/n)\mu_i$.

This tight example demonstrates that constraining each agent to receive at most one chore per block is insufficient by itself to achieve a better approximation guarantee.
Furthermore, this hard instance reveals that to achieve a tighter upper bound, the mechanism must prevent an agent from simultaneously receiving her largest chores and a large portion of other chores. 
To beat the $(2 - 1/n)$ approximation bound, the mechanism must ensure a structural trade-off: an agent who receives several smaller chores should avoid her largest ones, while an agent assigned a large chore should receive fewer remaining chores. 
A natural approach to achieving this balance is to allow agents to first nominate a set of their most costly chores.

We present the nomination mechanism and establish its approximation guarantees. 
For $n \ge 3$, fix an integer $k \ge 1$. 
We add zero-cost dummy chores so that $b = (|M|-k)/(n-k)$ is a positive integer.
For each agent $i\in N$, the nominated set is $H_{i,[k]}$, the set of her $k$ most costly chores. For each chore $e \in M$, we define $l(e)=\bigl|\{i\in N:e\in H_{i,[k]}\}\bigr|$ as the number of agents who nominated $e$.
By definition, we have $\sum_{e\in M}l(e)=nk$.
For each agent $i \in N$, we order the chores $M \setminus H_{i,[k]}$ in non-increasing order with respect to their cost under $c_i$ and split it into consecutive blocks $B_{i,1}, \dots, B_{i,b}$, each containing $n-k$ chores.
An allocation $\bX=\left(X_1,\ldots,X_n\right)$ is called \emph{safe} if, for every agent $i \in N$, either (1) $|X_i| \le 1$ or (2) $X_i\cap H_{i,[k]}=\emptyset$ and $|X_i\cap B_{i,r}|\le 1$ for every $r\in[b]$.

Intuitively, under a safe allocation, each agent either receives at most one chore in total (which trivially ensures MMS), or completely avoids her nominated large chores while receiving at most one chore from each remaining block.
To construct the nomination mechanism, we use the following extension of the matching decomposition used above, which allows vertices on one side to have smaller degrees.
The proof of the lemma is omitted as it follows from a standard application of Hall's Theorem~\cite{journals/jlms/Hall35}. 
The graph can be shown to admit an $R$-saturating matching, and the lemma then follows by removing this matching and applying induction.

\begin{lemma}\label{lemma:bipartite-decomposition}
    Let $G=(L\cup R,E)$ be a bipartite multigraph and $d\ge1$ be an integer.
    If every vertex in $L$ has degree at most $d$ and every vertex in $R$ has degree exactly $d$, then $E$ can be decomposed into $d$ matchings that saturate $R$.
\end{lemma}
    

\paragraph{$k$-Nomination Mechanism.}
Mechanism~\ref{alg:k-nomination} proceeds as follows.
First, each agent nominates her $k$ largest chores under her reported cost function, and her remaining chores are partitioned into blocks of size $n-k$.
We sample a subset $S \subseteq N$ of $n-k$ agents uniformly at random.
Following a similar construction as in Mechanism~\ref{alg:zero-nomination}, we build a bipartite graph $G=(L \cup R, E)$, where the set of chores constitutes the left side $L=M$, and the residual blocks of agents in $S$ form the right side $R=\{B_{i,r} : i \in S, r \in [b]\}$\footnote{Recall that these blocks exclude the nominated chores for agent $i$.}.
An edge $(e, B_{i,r})$ exists if and only if $e \in B_{i,r}$.
By Lemma~\ref{lemma:bipartite-decomposition}, $E$ can be decomposed into $n-k$ matchings that saturate $R$.
The mechanism selects one such matching uniformly at random and allocates each matched chore to the agent associated with the corresponding block.
Since each agent in $S$ holds $b$ blocks, this matching allocates exactly $b(n-k) = |M|-k$ chores, leaving $k$ chores unallocated and $k$ agents unselected.
Finally, the mechanism assigns these $k$ remaining chores to these $k$ agents via a uniformly random bijection.

\begin{algorithm}[htp]
    \caption{$k$-Nomination Mechanism}
    \label{alg:k-nomination}
    \KwIn{Reported cost functions $\bc$ and an integer $k \geq 1$}
    Add dummy chores to $M$ as described above\;
    $b\gets(|M|-k)/(n-k)$\;
    \For{each $i\in N$}{
        Nominate the chores in $H_{i,[k]}$\;
        Construct the blocks $B_{i,1},\ldots,B_{i,b}$, each of size $n-k$ from the remaining chores\;
    }
    Sample $S\subseteq N$ uniformly with $|S|=n-k$\;
    $L\gets M$, $R\gets\{B_{i,r}:i\in S,\ r\in[b]\}$\;
    $E\gets\{(e,B_{i,r}):i\in S,\ r\in[b],\ e\in B_{i,r}\}$\;
    Decompose $E=E_1\cup\cdots\cup E_{n-k}$ into $R$-saturating matchings; \tcp{By Lemma~\ref{lemma:bipartite-decomposition}}
    Sample $E'$ uniformly from $\{E_1,\ldots,E_{n-k}\}$\;
    $X_i\gets\{e\in M:(e,B_{i,r})\in E'\text{ for some }r\in[b]\}$ for each $i\in S$\;
    Assign $M\setminus\bigcup_{i\in S}X_i$ to $N\setminus S$ by a uniformly random bijection\;
    Remove all dummy chores from $\bX$\;
    \KwOut{An allocation $\bX$}
\end{algorithm}

\medskip

To illustrate the execution of Mechanism~\ref{alg:k-nomination}, consider the following example.
\begin{example}\label{example: mechanism2}
    Let $n=3$, $k=1$, and $M=\{e_1,\ldots,e_5\}$, which yields $b=2$ blocks per agent without requiring dummy chores.
    Suppose agents $1$ and $2$ nominate $e_1$ and $e_2$, respectively, with their residual blocks given by
    $B_{1,1}=\{e_2,e_3\}$, $B_{1,2} =\{e_4,e_5\}$, $B_{2,1}=\{e_1,e_3\}$, $B_{2,2}=\{e_4,e_5\}$.
    Assume that the mechanism selects $S=\{1,2\}$.
    Figure~\ref{fig:nomination-graph} illustrates a bipartite decomposition into two matchings, $E_1$ and $E_2$.
    If $E_1$ is sampled, agents $1$ and $2$ receive $X_1=\{e_2,e_4\}$ and $X_2=\{e_3,e_5\}$, respectively, while agent $3$ receives the remaining chore, which yields $X_3=\{e_1\}$.
    Otherwise, $E_2$ is sampled, the resulting allocation is $X_1=\{e_3,e_5\}$, $X_2=\{e_1,e_4\}$, and $X_3=\{e_2\}$.
    Under this decomposition, each allocation occurs with probability $1/2$.
\end{example}

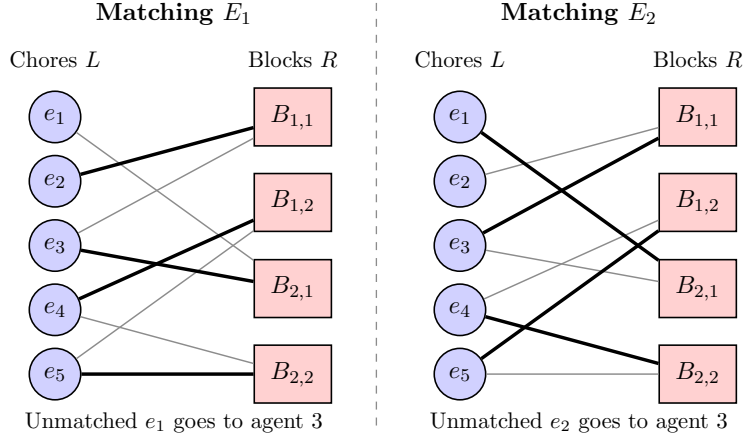
\begin{figure}[H]
    \centering
    \begin{tikzpicture}[
        scale=0.85,
        transform shape,
        chore/.style={circle,draw=black,line width=0.6pt,fill=blue!18!white,minimum size=8mm,inner sep=0pt},
        block/.style={rectangle,draw=black,line width=0.6pt,fill=red!18!white,minimum width=12mm,minimum height=9mm,inner sep=1pt},
        ordinary edge/.style={draw=black!45,line width=0.55pt},
        matching edge/.style={draw=black,line width=1.3pt}
    ]
        \draw[black!45,dashed,line width=0.5pt] (5.6,-0.8) -- (5.6,5.9);
        \foreach \panel/\shift in {1/0,2/6.3} {
            \begin{scope}[xshift=\shift cm]
                \node[font=\bfseries] at (2.45,5.6) {Matching $E_{\panel}$};
                \node[font=\small] at (0.6,4.9) {Chores $L$};
                \node[font=\small] at (4.3,4.9) {Blocks $R$};
                \foreach \j in {1,...,5} {
                    \node[chore] (e\panel-\j) at (0.6,{5-\j}) {$e_\j$};
                }
                \node[block] (b\panel-11) at (4.3,4) {$B_{1,1}$};
                \node[block] (b\panel-12) at (4.3,2.67) {$B_{1,2}$};
                \node[block] (b\panel-21) at (4.3,1.33) {$B_{2,1}$};
                \node[block] (b\panel-22) at (4.3,0) {$B_{2,2}$};
                \foreach \e/\block in {2/11,3/11,4/12,5/12,1/21,3/21,4/22,5/22} {
                    \draw[ordinary edge] (e\panel-\e) -- (b\panel-\block);
                }
                \ifnum\panel=1
                    \foreach \e/\block in {2/11,4/12,3/21,5/22} {
                        \draw[matching edge] (e\panel-\e) -- (b\panel-\block);
                    }
                    \node[font=\small] at (2.45,-0.75) {Unmatched $e_1$ goes to agent $3$};
                \else
                    \foreach \e/\block in {3/11,5/12,1/21,4/22} {
                        \draw[matching edge] (e\panel-\e) -- (b\panel-\block);
                    }
                    \node[font=\small] at (2.45,-0.75) {Unmatched $e_2$ goes to agent $3$};
                \fi
            \end{scope}
        }
    \end{tikzpicture}
    \caption{Illustration of Example~\ref{example: mechanism2}.}
    \label{fig:nomination-graph}
\end{figure}

Next, we show that the realization of Mechanism~\ref{alg:k-nomination} is safe and determine the marginal probabilities of chore assignment.

\begin{lemma}\label{lem:safe-implementation}
    Every realization of Mechanism~\ref{alg:k-nomination} is safe.
    Furthermore, for any reported cost profile, agent $i \in N$, and chore $e \in M$, the probability that agent $i$ receives $e$ is given by
\begin{equation*}
\Pr[e\in X_i] =\frac{l(e)-1+n \cdot \mathbf{1}(e\notin H_{i,[k]})}{n\cdot\left(n-1\right)}.
\end{equation*}
\end{lemma}
\begin{proof}
    We first establish that every realization yields a safe allocation. Fix an execution of the mechanism with a selected subset $S$ and a sampled matching $E^*$. By construction, $E^*$ saturates all $b(n-k)$ block vertices, allocating $b(n-k)$ distinct chores to the agents in $S$. The remaining $k$ chores are distributed to the $k$ agents outside $S$.
    Consequently, each agent in $S$ receives one chore from each of her residual blocks, while each agent outside $S$ receives exactly one chore. 
    
    We next analyze the marginal probability of chore assignment under Mechanism~\ref{alg:k-nomination}. Fix a reported cost profile, an agent $i \in N$, and a chore $e \in M$. 
    To calculate the probability that $i$ receives $e$, i.e., $\Pr[e \in X_i]$, we consider whether agent $i$ belongs to the sampled subset $S$.
    \begin{itemize}
        \item \textbf{Case 1: $i \in S$.} If agent $i$ did not nominate chore $e$, the unique edge between $e$ and one of $i$'s blocks appears in exactly one of the $n-k$ decomposition matchings. If $i$ nominated $e$, no such edge exists. Since $\Pr[i \in S] = (n-k)/n$, the probability that $i$ receives $e$ via matching decomposition is $\Pr[e \in X_i, \, i \in S] = \frac{1}{n}\cdot \mathbf{1}(e \notin H_{i,[k]})$.
    
        \item \textbf{Case 2: $i \notin S$.} 
        In this case, we first compute the probability that chore $e$ remains unmatched. 
        Let $l'(e)$ denote the number of agents in $N \setminus \{i\}$ who nominate chore $e$. 
        For any realization of $S$ that satisfies $i \notin S$, the degree of vertex $e$ in the bipartite graph is precisely the number of agents in $S$ who did not nominate $e$.
        Hence, the matching leaves $e$ unmatched with conditional probability
        \begin{equation*}
\Pr[e \text{ remains unmatched} \mid S] = \frac{\bigl|\{j \in S : e \in H_{j,[k]}\}\bigr|}{n-k}.
\end{equation*}
        We now take the expectation over all realizations of $S$ conditional on $i \notin S$. Under this condition, each agent $j \neq i$ belongs to $S$ with probability $(n-k)/(n-1)$, which yields
        \begin{equation*}
\begin{aligned}
            \Pr[e \text{ remains unmatched} \mid i \notin S] 
            &= \E[\Pr[e \text{ remains unmatched} \mid S] \mid i \notin S]\\
            &= \frac{1}{n-k} \cdot \, \E\!\left[ \bigl|\{j \in S : e \in H_{j,[k]}\}\bigr| \,\middle|\, i \notin S \right] \\
            &= \frac{1}{n-k} \cdot \frac{l'(e)(n-k)}{n-1} = \frac{l'(e)}{n-1}.
        \end{aligned}
\end{equation*}
        In this event, chore $e$ is allocated to agent $i$ via the random bijection with probability $1/k$. 
        Together with $\Pr[i \notin S] = k/n$, we obtain $\Pr[e \in X_i, \, i \notin S] = \frac{l'(e)}{n\cdot\left(n-1\right)}$.
    \end{itemize}
    Combining both cases and applying $l'(e) = l(e) - \mathbf{1}(e \in H_{i,[k]})$ yields the lemma.
\end{proof}

Given the safety properties and marginal probabilities provided by Lemma~\ref{lem:safe-implementation}, it remains to prove the truthfulness and approximate MMS guarantees of Mechanism~\ref{alg:k-nomination}.

\begin{theorem}\label{thm:top-avoidance}
    For any $n \ge 3$ and $1 \le k \le \lfloor(n-1)/2\rfloor$, Mechanism~\ref{alg:k-nomination} runs in polynomial time and returns an allocation that is TIE, ex-ante EF, and ex-post $\left(2 - 1/(n-k)\right)$-MMS.
\end{theorem}
\begin{proof}
    We first establish that Mechanism~\ref{alg:k-nomination} is TIE. Fix an agent $i \in N$, her true cost function $c_i$, and the reports of all other agents. 
    By Lemma~\ref{lem:safe-implementation}, the expected cost incurred by agent $i$ is
    \begin{align*}
        \E[c_i(X_i)]
        &= \sum_{e \in M} c_i(e) \left( \frac{\mathbf{1}(e \notin H_{i,[k]})}{n} + \frac{l'(e)}{n\cdot\left(n-1\right)} \right) \\
        &= \frac{c_i(M)}{n} - \frac{1}{n}\cdot c_i\!\left(H_{i,[k]}\right) + \frac{1}{n\cdot\left(n-1\right)} \sum_{e \in M} l'(e) c_i(e).
    \end{align*}
    In the above expansion, only the term $-\frac{1}{n}c_i\!\left(H_{i,[k]}\right)$ depends on agent $i$'s report. 
    Consequently, $\E[c_i(X_i)]$ is minimized by nominating her $k$ most costly chores with respect to $c_i$, which corresponds to truthful reporting.
    Furthermore, under truthful reporting, for every $j\in N$,
    \begin{equation*}
    \begin{aligned}
        \E[c_i(X_j)]-\E[c_i(X_i)]
        =\sum_{e\in M}c_i(e)
          \bigl(\Pr[e\in X_j]-\Pr[e\in X_i]\bigr)
        =\frac{c_i(H_{i,[k]})-c_i(H_{j,[k]})}{n-1}
        \geq 0,
    \end{aligned}
    \end{equation*}
    where the second equality is by Lemma~\ref{lem:safe-implementation} and the inequality holds because $H_{i,[k]}$ contains the $k$
    most costly chores according to $c_i$ whereas $H_{j,[k]}$ is another
    set of $k$ chores. Hence,
    $\E[c_i(X_i)]\leq \E[c_i(X_j)]$ for every $j\in N$, and the mechanism
    is ex-ante EF.

    Next, we prove the ex-post MMS approximation guarantee. 
    Fix an agent $i \in N$ and a realized allocation $\bX$. 
    If $|X_i| \le 1$, then $c_i(X_i) \le \mu_i$ holds immediately. 
    Now suppose that $|X_i| \ge 2$. 
    The safety condition guarantees that agent $i$ receives no nominated chore and at most one chore from each block $B_{i,r}$. 
    For each $r\in[b]$, let $y_r=c_{i,\left(k+(r-1)(n-k)+1\right)}$ be the largest cost of chores in $B_{i,r}$. Then $c_i\!\left(X_i\right)\le\sum_{r=1}^b y_r$.
    By construction, each nominated chore costs at least $y_1$. 
    Therefore, we have
    \begin{equation*}
        c_i(M) \ge c_i\!\left(H_{i,[k]}\right) + y_1 + (n-k) \sum_{r=2}^b y_r \ge (k+1)y_1 + (n-k) \sum_{r=2}^b y_r.
    \end{equation*}
    Consequently, it follows that $\sum_{r=2}^b y_r \le (c_i(M) - (k+1)\cdot y_1)/(n-k)$. 
    Combined with the bounds $c_i(M) \le n\cdot \mu_i$, $y_1 \le \mu_i$, and $n - 2k - 1 \ge 0$, we obtain
    \begin{equation}
c_i(X_i) \le y_1 + \frac{c_i(M) - (k+1) \cdot y_1}{n-k} = \frac{c_i(M) + (n-2k-1)\cdot y_1}{n-k} \le \left(2 - \frac{1}{n-k}\right) \mu_i. \label{eq:ratio-k-nomination}
\end{equation}

    Finally, we analyze the computational complexity of Mechanism~\ref{alg:k-nomination}. 
    Because the mechanism introduces at most $n$ dummy chores, sorting the chores and constructing the regularized block graph both run in polynomial time. The $n-k$ matchings can be efficiently computed via successive maximum bipartite matching algorithms. Furthermore, sampling $S$, the matching, and the random bijection relies on uniform integer sampling, which runs in expected polynomial time using standard rejection sampling.
\end{proof}

Notice that the approximation ratio decreases in $k$.
Given the upper bound of $k \le \lfloor(n-1)/2\rfloor$, the minimum ratio $2-\frac{1}{\lceil(n+1)/2\rceil}$ is achieved when $k = \lfloor(n-1)/2\rfloor$.
However, the construction above applies only when $n \ge 3$ as we have $k=0$ when $n=2$.
For two agents, we modify the mechanism such that each agent nominates one chore (i.e., $k=1$).
Depending on whether the two nominations coincide, we partition and balance the remaining chores between the two agents. 
In Appendix~\ref{sec:nomination-two}, we prove that under a suitably constructed partition, this mechanism achieves TIE and an ex-post $4/3$-MMS guarantee, matching the lower bound for two agents in Section~\ref{sec:lower_bounds} and showing that the TIE mechanism is optimal.

As a final remark, our $k$-nomination mechanism fits into the proper scoring rule framework of Freeman, Witkowski, Vaughan, and Pennock~\cite{journals/mansci/FreemanWVP24}.
Specifically, Freeman et al. provide an equivalence between the truthful fair division mechanisms of divisible resources and the wagering mechanisms.
A class of truthful mechanisms is then identified.
The probabilities in Lemma~\ref{lem:safe-implementation} yield a fractional division rule that belongs to this class.
Moreover, this class of truthful mechanisms characterized by Freeman et al. also guarantees envy-freeness.

\subsection{Limitations of the Nomination Mechanism}

In this subsection, we discuss the limitations of the $k$-nomination mechanism. 
We first clarify the requirement $k \le \lfloor(n-1)/2\rfloor$ in Theorem~\ref{thm:top-avoidance}. 
While the mechanism remains well-defined and TIE for any $1 \le k < n$, this restriction on $k$ is essential to secure the claimed approximation guarantee. 
Specifically, for the last inequality of~\Cref{eq:ratio-k-nomination} to hold, we need the coefficient of $y_1$ to be non-negative, which requires $k \le \lfloor(n-1)/2\rfloor$.
This performance loss is inherent to the mechanism: an increase in $k$ leaves fewer agents to share the chores that remain after the assignment of the $k$ singleton bundles. 
To illustrate this, consider an instance with $m = k + b(n-k)$ chores, where each chore incurs a cost of $1$ for all agents, for which we have $\mu_i = \lceil m/n \rceil$ for every agent $i$. 
Following Mechanism~\ref{alg:k-nomination}, each agent forms exactly $b$ blocks of size $n - k$. 
In every realization, each of the $k$ agents outside $S$ receives a single chore, whereas each of the $n - k$ agents in $S$ receives one chore from each of her $b$ blocks. 
Consequently, every agent $i \in S$ incurs a total cost of $b$. 
As $b$ increases, her approximation ratio satisfies
\begin{equation*}
\frac{c_i(X_i)}{\mu_i} = \frac{b}{\lceil (k + b(n-k))/n \rceil} \xrightarrow[b \to \infty]{} \frac{n}{n-k},
\end{equation*}
which is strictly larger than $2$ when $k>n/2$, and approaches $n$ when $k \to n-1$.

Within the valid range $k \le \lfloor(n-1)/2\rfloor$, we next establish that the approximation ratio of the $k$-nomination mechanism is tight and cannot be further improved.

\begin{lemma} \label{prop:nomination-tight}
    Fix $n \ge 3$ and $1 \le k \le \lfloor (n-1)/2 \rfloor$. For every $\varepsilon > 0$, there exists an instance where every realization of the $k$-nomination mechanism admits some agent $i \in N$ with
    \begin{equation*}
c_i(X_i) \ge \left(2 - \frac{1}{n-k} - \varepsilon\right)\mu_i.
\end{equation*}
\end{lemma}
\begin{proof}
    Fix an integer $t \ge 1$. 
    Consider an instance in which all agents share an identical cost function, and the total cost of all chores is $c_i(M) = n$ for each agent $i \in N$. 
    Let there be $k+1$ large chores, each with cost $1$, and $(n-k-1)(1+t(n-k))$ small chores, each with cost $1/(1+t(n-k))$.
    Consequently, we have $\mu_i = 1$ for all $i \in N$.
    Under the nomination mechanism, each agent nominates $k$ of the $k+1$ large chores. 
    For every agent, the sequence of non-nominated chores thus comprises a single unnominated large chore followed by all small chores. 
    Therefore, some agent $i \in S$ receives a large chore together with one small chore from each of her remaining $b-1$ blocks. 
    Consequently, the total cost incurred by agent $i$ is
    \begin{equation*}
c_i(X_i) = 1 + \frac{t(n-k-1)}{1+t(n-k)} = 2 - \frac{1}{n-k} - \frac{n-k-1}{(n-k)(1+t(n-k))} \xrightarrow[t \to \infty]{} 2 - \frac{1}{n-k},
\end{equation*}
    Therefore, when $t$ is chosen to be sufficiently large, we have $c_i(X_i) \ge 2 - 1/(n-k) - \varepsilon$ for any given $\varepsilon > 0$.
\end{proof}

\section{\texorpdfstring{$1.97$}{}-MMS Mechanism}
\label{sec:2_minus_eps_overview}
Although the nomination mechanism in \Cref{sec:nomination_overview} improves the approximation ratio, its guarantee approaches $2$ as $n\to\infty$.
Therefore, in this section, we introduce a new approach whose MMS approximation guarantee is bounded away from $2$ by a fixed constant for all sufficiently large $n$ (Theorem~\ref{thm:1.97}).

\subsection{Sorted Slot Rounding with Unequal Fractions}
\label{subsec:sorted-slots}

Our mechanism will use a subroutine that adjusts the uniform fractional allocation to produce several intermediate fractional allocations.
We first explain how to round a general fractional allocation into a distribution over integral allocations via sorted slots.

The mechanism in \Cref{sec:2-1/n} can be viewed as rounding the \emph{uniform fractional allocation}, in which each agent is assigned a fraction of $1/n$ of every chore.
For each agent, order the chores by nonincreasing cost, using a fixed public order for tie-breaking, and partition them into slots of $n$ chores.
Each full slot has total fraction $1$, and the mechanism rounds the fractional allocation while giving each agent at most one chore through each slot.

In the following, we show that this idea can be extended to unequal fractions.
Given a fractional allocation, in which the nonnegative fractions of each chore sum to one across agents, group each agent's positive chore fractions in cost order into slots of total fraction $1$, except possibly the last slot.
If a chore fraction crosses a slot boundary, split it between the two slots.
This fractional allocation can be rounded into a distribution over complete integral allocations that preserves the assigned fractions in expectation and gives each agent at most one chore through each slot (see Figure~\ref{fig:sorted-slot-example} for the illustration).

We need this extension because our plan is to construct several intermediate fractional allocations whose average is the uniform fractional allocation.
More precisely, let $R$ be the number of intermediate fractional allocations in the construction under consideration. For each $s\in[R]$, write $x^{(s)}=(x^{(s)}_{ie})_{i\in N,e\in M}$ for the $s$-th allocation, where $x^{(s)}_{ie}$ is the fraction of chore $e$ assigned to agent $i$ in that allocation.
To preserve the equal-marginal property, we require
\begin{equation*}
\frac1R \cdot \sum_{s=1}^R x^{(s)}_{ie}=\frac1n,\qquad \forall i\in N,\ e\in M.
\end{equation*}
We choose $s$ uniformly from $[R]$ and round only $x^{(s)}$. Thus $x^{(s)}_{ie}$ is also the conditional probability that agent $i$ receives chore $e$, given the choice of $s$.
We will design them so that every integral allocation produced by rounding any of them satisfies the desired $1.97$-MMS guarantee.
Choosing one intermediate fractional allocation uniformly at random and then rounding it preserves this guarantee and gives every agent each chore with probability $1/n$.

\begin{example}
Suppose $m=70$ and $n=10$, and let $e_1,\ldots,e_{70}$ be ordered by nonincreasing cost for an arbitrary agent $i$.
The uniform fractional allocation gives the slot partition
\begin{equation*}
\{e_1,\ldots,e_{10}\},\{e_{11},\ldots,e_{20}\},\ldots,\{e_{61},\ldots,e_{70}\}.
\end{equation*}
If agent $i$ is instead assigned a fraction of $0.2$ of each of the first ten chores and a fraction of $1/12$ of each remaining chore, then the partition is
\begin{equation*}
\{e_1,\ldots,e_5\},\{e_6,\ldots,e_{10}\},\{e_{11},\ldots,e_{22}\},\{e_{23},\ldots,e_{34}\},\ldots,\{e_{59},\ldots,e_{70}\}.
\end{equation*}
In both cases, every slot has total fraction $1$, so rounding gives agent $i$ at most one chore through each slot.
Figure~\ref{fig:sorted-slot-example} illustrates these two slot partitions.
\end{example}
\begin{figure}[htbp]
\centering
\begin{tikzpicture}[x=1.9cm,y=1cm,font=\small]
  \node[anchor=west] at (0,1.35) {(a) Uniform fractions: $x_{ie}=0.1$ for all $e\in M$};
  \fill[orange!25] (0,0) rectangle (1,0.55);
  \fill[blue!8] (1,0) rectangle (7,0.55);
  \foreach \j in {1,...,69} {
    \draw[black!25,thin] ({\j/10},0) -- ({\j/10},0.55);
  }
  \draw[thick] (0,0) rectangle (7,0.55);
  \foreach \s in {1,...,6} {\draw[thick] (\s,0) -- (\s,0.55);}
  \foreach \s/\a/\b in {1/1/10,2/11/20,3/21/30,4/31/40,5/41/50,6/51/60,7/61/70} {
    \node at ({\s-0.5},0.86) {$S_{\s}$};
    \node[font=\scriptsize] at ({\s-0.5},-0.25) {$e_{\a},\ldots,e_{\b}$};
  }

  \node[anchor=west] at (0,-1.05) {(b) Unequal fractions: $x_{ie}=0.2$ for $e\in\{e_1,\ldots,e_{10}\}$; $x_{ie}=1/12$ otherwise};
  \fill[orange!25] (0,-2.4) rectangle (2,-1.85);
  \fill[blue!8] (2,-2.4) rectangle (7,-1.85);
  \foreach \j in {1,...,9} {
    \draw[black!25,thin] ({\j/5},-2.4) -- ({\j/5},-1.85);
  }
  \foreach \j in {1,...,59} {
    \draw[black!25,thin] ({2+\j/12},-2.4) -- ({2+\j/12},-1.85);
  }
  \draw[thick] (0,-2.4) rectangle (7,-1.85);
  \foreach \s in {1,...,6} {\draw[thick] (\s,-2.4) -- (\s,-1.85);}
  \foreach \s/\a/\b in {1/1/5,2/6/10,3/11/22,4/23/34,5/35/46,6/47/58,7/59/70} {
    \node at ({\s-0.5},-1.54) {$S_{\s}$};
    \node[font=\scriptsize] at ({\s-0.5},-2.65) {$e_{\a},\ldots,e_{\b}$};
  }
  \draw[-{Stealth},semithick] (0,-3.18) -- (7,-3.18)
    node[midway,below] {Chores ordered by nonincreasing cost for agent $i$};
\end{tikzpicture}
\caption{Sorted slots for the example with $m=70$ and $n=10$.}
\label{fig:sorted-slot-example}
\end{figure}
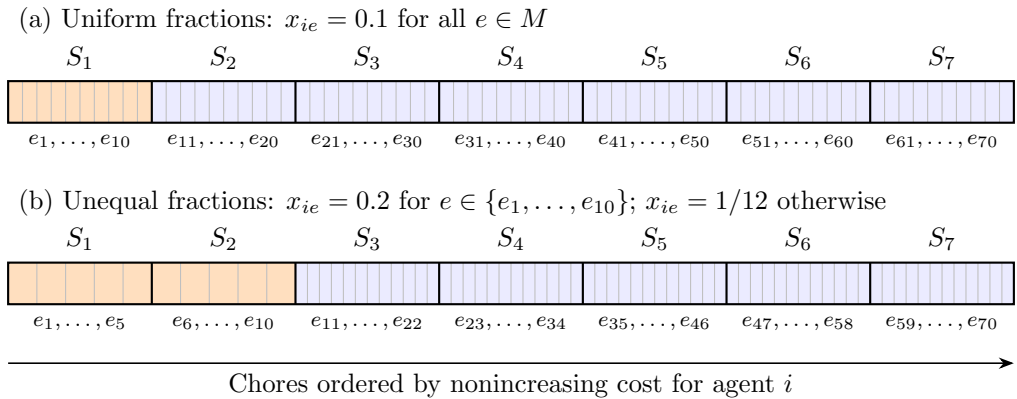

Thus, changing the fractions changes the slot structure while retaining the guarantee of at most one chore through each slot after rounding.
To bound an agent's cost, we bound the first slot's contribution by the largest cost among chores of positive fraction, and each later slot's contribution by the fractional cost of the preceding full slot.
The following lemma formalizes this guarantee and a useful refinement.


\begin{lemma}[Sorted-slot rounding]\label{lem:sorted-slot-rounding}
Let $\mathbf{x}=(x_{ie})_{i\in N,e\in M}$ be a feasible fractional allocation.
There exists a decomposition of $\mathbf{x}$ into a distribution over integral allocations such that every allocation $\bX$ in the support satisfies
\begin{equation*}
c_i(X_i)\le \max_{e:x_{ie}>0}\left\{c_i(e)\right\} +\underbrace{\sum_{e\in M}x_{ie}\cdot c_i(e)}_{\text{fractional cost}}, \qquad \forall i\in N.
\end{equation*}
Moreover, if $|M|>n$ and $\sum_{e\in H_{i,[n]}}x_{ie}=1$, the upper bound can be strengthened to
\begin{equation*}
c_i(X_i)\le c_{i,(1)}+c_{i,(n+1)} +\sum_{e\in M\setminus H_{i,[n]}}x_{ie}\cdot c_i(e).
\end{equation*}
\end{lemma}

The proof of Lemma~\ref{lem:sorted-slot-rounding} is given in Appendix~\ref{subsec:rounding}.
In the following, we refer to the two upper bounds as the general bound and the strengthened bound (which further requires $\sum_{e\in H_{i,[n]}}x_{ie}=1$).
Note that for the uniform fractional allocation, the general bound becomes
\begin{equation*}
c_i(X_i)\le c_{i,(1)}+{c_i(M)}/n~.
\end{equation*}

\subsection{Easy and Hard Agents}
\label{subsec:intuitions}
In this subsection, we first examine the guarantee obtained by rounding the uniform
fractional allocation using \Cref{lem:sorted-slot-rounding}.
This allows us to identify the agents whose fractional shares
need to be adjusted to achieve the desired approximation guarantee.
In the remainder of this section, unless specified otherwise, we fix an arbitrary agent $i \in N$.
We further define
\begin{equation*}
L_i=\max\left\{{c_i(M)}/n,\ c_{i,(1)},\ 2c_{i,(n+1)}\right\}.
\end{equation*}
We remark that each term is a lower bound on $\mu_i$.
Thus, we have $0 < L_i\le\mu_i$.
By normalizing each cost function $c_i$ by $L_i$, we assume without loss of generality that $L_i = 1$ for all $i \in N$.
Equivalently, we have
\begin{equation*}
c_i\!\left(M\right)\le n,\qquad c_{i,(1)}\le1,\qquad c_{i,(n+1)}\le0.5.
\end{equation*}

For the uniform fractional allocation $x_{ie}=1/n$, the fractions on $H_{i,[n]}$ sum to $1$.
Thus, the strengthened bound in \Cref{lem:sorted-slot-rounding} applies.
Using $x_{ie}=1/n$ and denoting the resulting bound by $\costbound{i}$, we obtain
\begin{equation*}
c_i(X_i)
 \le \underbrace{c_{i,(1)}}_{\text{largest cost}}
    +\underbrace{c_{i,(n+1)}}_{\substack{\text{largest cost}\\\text{outside }H_{i,[n]}}}
    +\underbrace{\frac{1}{n}\cdot c_i(M\setminus H_{i,[n]})}_{\substack{\text{fractional cost}\\\text{outside }H_{i,[n]}}}
 =:\costbound{i}.
\end{equation*}
Since $L_i = 1 \le \mu_i$, the bound $c_i(X_i) \le \costbound{i}$ directly implies an MMS approximation ratio of $\costbound{i}$.
To write $\costbound{i}$ as a function of the total cost, we further define
\begin{equation*}
\begin{aligned}
  \excess{i}:=c_i\!\left(H_{i,[n]}\right)-n\cdot c_{i,(n+1)}
      =\sum_{e\in M}\left(c_i\!\left(e\right)-c_{i,(n+1)}\right)^+,
 \text{ and }
  \costbound{i}=c_{i,(1)}+\frac{c_i(M)-\excess{i}}{n},
 \end{aligned}
\end{equation*}

\Cref{fig:areas} illustrates these quantities on the sorted cost profile.
Specifically, the red area is $\excess{i}$, and removing it leaves the blue area $c_i(M)-\excess{i}=n\cdot c_{i,(n+1)}+c_i(M\setminus H_{i,[n]})$.
Dividing the blue area by $n$ and adding $c_{i,(1)}$ gives $\costbound{i}$.

\begin{figure}[htbp]
\centering
\begin{tikzpicture}
[x=0.87cm,y=4.05cm,>=Latex,
every node/.style={font=\small}]
 \def\base{0.44}
 \foreach \x/\h in {1/0.98,2.15/0.86,3.30/0.76,4.45/0.67,
                    5.60/0.59,6.75/0.51,7.90/0.44,9.05/0.37,
                    10.20/0.29,11.35/0.21,12.50/0.13} {
   \pgfmathsetmacro{\lowerheight}{min(\h,\base)}
   \fill[blue!25] (\x,0) rectangle ({\x+1},\lowerheight);
   \ifdim\h pt>\base pt
     \fill[red!33] (\x,\base) rectangle ({\x+1},\h);
   \fi
   \draw[black!70,line width=0.35pt] (\x,0) rectangle ({\x+1},\h);
 }
 \draw[line width=0.9pt] (1,0) rectangle (2,0.98);
 \draw[line width=0.9pt] (7.90,0) rectangle (8.90,0.44);
 \draw[->] (0.65,0) -- (14.1,0);
 \draw[densely dotted,line width=0.8pt] (0.62,\base) -- (14.15,\base);
 \draw[densely dotted,black!55] (0.47,0.98) -- (1,0.98);
 \draw[<->,line width=0.65pt] (0.52,0) --
     node[left=3pt] {$c_{i,(1)}$} (0.52,0.98);
 \draw[<->,line width=0.65pt] (14.0,0) --
     node[right=3pt] {$c_{i,(n+1)}$} (14.0,\base);
 \fill (8.4,\base) circle[radius=1.35pt];
 \node[above right,align=left] at (9.5,0.6)
    {$(n+1)$-th chore:\\height $c_{i,(n+1)}$};
 \draw[->,black!65] (9.6,0.61) -- (8.45,0.455);
 \node[text=red!75!black,align=center] at (8.5,1)
    {Excess above $c_{i,(n+1)}$\\red area $=\excess{i}$};
 \draw[->,red!75!black] (7.0,0.91) -- (5.50,0.69);
 \foreach \x/\lab in {1.5/1,2.65/2,3.8/3,4.95/\cdots,
                       6.1/n-1,7.25/n,8.4/n+1,
                       9.55/n+2,10.7/\cdots,11.85/m-1,13/m} {
   \node[below=3pt] at (\x,0) {$\lab$};
 }
 \draw[black!65] (1,-0.13) -- (1,-0.18) -- (7.75,-0.18) -- (7.75,-0.13);
 \node[below=4pt] at (4.375,-0.18) {$H_{i,[n]}$: $n$ most costly chores};
 \draw[black!65] (7.90,-0.13) -- (7.90,-0.18) -- (13.50,-0.18) -- (13.50,-0.13);
 \node[below=4pt] at (10.70,-0.18) {remaining chores};
\end{tikzpicture}
\caption{Area interpretation of $\costbound{i}$. Each bar has height equal to the chore's cost. The red area above $c_{i,(n+1)}$ is $\excess{i}$, and the blue area is $n\cdot c_{i,(n+1)}+c_i(M\setminus H_{i,[n]})=c_i(M)-\excess{i}$.}
\label{fig:areas}
\end{figure}
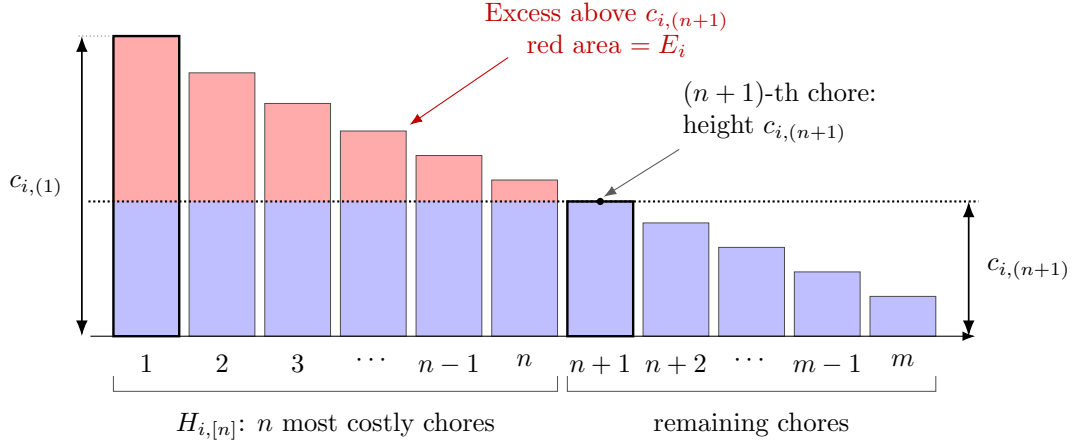

Since $c_{i,(1)}\le1$ and $(c_i(M)-\excess{i})/n\le c_i(M)/n\le1$, we have $\costbound{i}\le2$.
Intuitively, for a fixed total cost and largest chore cost, a larger excess $\excess{i}$ yields a tighter bound $\costbound{i}$.
Conversely, for $\costbound{i}$ to approach $2$, $c_{i,(1)}$ must be close to $1$ and $\excess{i}/n$ must be small.
This structure naturally motivates us to partition the agents into two groups based on $\costbound{i}$.

\paragraph{Easy and Hard Agents.}
Throughout the analysis, we fix $\delta = 0.05$.
We classify an agent $i$ as \emph{easy} if $\costbound{i} \le 2 - \delta = 1.95$, and \emph{hard} otherwise.
Let $Q = \{i : \costbound{i} > 1.95\}$ denote the set of hard agents.
Thus, every easy agent exhibits a slack of $0.02$ relative to the target threshold of $1.97$.
If all agents are easy, simply rounding the uniform fractional allocation guarantees an ex-post $1.95$-MMS bound.
Hence, our main challenge is to establish an approximation ratio below $2$ for hard agents while maintaining equal marginals.

\medskip

To address the difficulty, we redistribute their fractions differently across several \emph{intermediate fractional allocations}, whose average under the sampling distribution remains the uniform allocation.
We then sample an intermediate allocation and round it using \Cref{lem:sorted-slot-rounding}.
Since rounding preserves the fractions in expectation, the final mechanism assigns each chore to each agent with probability $1/n$ at every report profile, which ensures TIE.
Therefore, our goal is to design these fractions so that every integral allocation produced by rounding any intermediate allocation satisfies a $1.97$-MMS guarantee for every agent.

In fact, the cost bound in \Cref{lem:sorted-slot-rounding} guides these adjustments: it consists of the largest cost among chores of positive fraction and the total fractional cost.
For each hard agent, we consider two complementary adjustments:
\begin{itemize}
    \item Increase her fractions of large chores and decrease those of smaller chores enough to reduce her total fractional cost, e.g., see Figure~\ref{fig:sorted-slot-example} (b).
    \item Set her fractions of these large chores to zero and increase those of smaller chores, which reduces the largest cost among chores of positive fraction.
\end{itemize}
We will use the structural properties of hard agents to coordinate these adjustments across agents, keeping each intermediate allocation feasible and their average equal to the original uniform fractional allocation.
The following example illustrates the two adjustments.

\begin{example}\label{example:intermediate-allocations}
Set $n=100$, and consider an agent $i$ whose cost function consists of one large chore of cost $1$ and $9{,}900$ small chores, each of cost $0.01$.
Her total cost is $c_i(M) = 100$, and her MMS is $\mu_i = 1$.
Under the uniform fractional allocation, agent $i$ receives a fractional cost of $1$.
Since $\excess{i} = 0.99$, her corresponding upper bound is $\costbound{i} = 1 + (100 - 0.99)/100 = 1.9901$.
Here, we observe that $c_i(H_{i,[25]}) = 1.24$ and $c_i(M \setminus H_{i,[25]}) = 98.76$.
Consider the two candidate fractional options for agent $i$ presented in \Cref{tab:two-intermediate}, each chosen with probability $0.5$.
\end{example}
\begin{table}[ht!]
    \centering
    \begin{tabular}{lcc}
\toprule
 & First intermediate allocation & Second intermediate allocation\\
\midrule
Each chore in $H_{i,[25]}$ & $4/200$ & $0$\\
Each chore outside $H_{i,[25]}$ & $1/200$ & $3/200$\\
\bottomrule
\end{tabular}
    \caption{Agent $i$'s fractions in two equally likely intermediate allocations.}
    \label{tab:two-intermediate}
\end{table}

In each row, the two fractions average to $1/100$, recovering the uniform fractional allocation for agent $i$.
In the first allocation, doubling the fractions in $H_{i,[25]}$ adds little fractional cost, while halving those outside $H_{i,[25]}$ removes much more.
The largest chore of positive fraction still has cost $1$, but the fractional cost decreases to $0.5186$.
By \Cref{lem:sorted-slot-rounding}, every integral allocation produced by rounding satisfies
\begin{equation*}
c_i(X_i)\le \underbrace{1}_{\substack{\text{largest cost among}\\\text{chores of positive fraction}}} +\underbrace{0.02\times 1.24+\frac{98.76}{200}}_{\text{fractional cost}} =1.5186~.
\end{equation*}
In the second allocation, all fractions in $H_{i,[25]}$ are zero, so the largest chore of positive fraction has cost only $0.01$.
This reduction compensates for the higher fractional cost of $1.4814$, giving
\begin{equation*}
c_i(X_i)\le \underbrace{0.01}_{\substack{\text{largest cost among}\\\text{chores of positive fraction}}} +\underbrace{\frac3{200}\times 98.76}_{\text{fractional cost}} =1.4914~.
\end{equation*}
Both bounds improve on $\costbound{i}=1.9901$, while the average fraction of each chore remains $0.01$.

\paragraph{Rebalancing for General Hard Agents.}
The example illustrates two complementary ways to obtain a rounding bound below $2$.
The structural properties of hard agents allow us to choose suitable sets of most costly chores and suitable fractions for the general construction.
For a hard agent, both terms $c_{i,(1)}$ and $(c_i(M)-\excess{i})/n$ must be close to $1$.
Moreover, the second term can be close to $1$ only if $\excess{i}/n$ is close to $0$.
That is, the red area in \Cref{fig:areas} must be small relative to $n$.
We make these bounds explicit:
For a hard agent $i$, we have $2-\costbound{i}<\delta$, which is
\begin{equation*}
(1-c_{i,(1)})+\left(1-\frac{c_i(M)}n\right)+\frac{\excess{i}}n<\delta.
\end{equation*}
Since all terms on the left-hand side are nonnegative, we have
\begin{equation}\label{eqn:hard-Ei-average-pi}
      \excess{i}<n\cdot (c_{i,(1)}-1+\delta)\le n\cdot\delta,
      \qquad
      \frac{c_i(M)}n>1-\delta+\frac{\excess{i}}n,
      \qquad c_{i,(1)}>1-\delta.
\end{equation}

For a hard agent $i$, the largest chore has cost close to $1$,
but only a limited number of chores can have similarly high costs,
since their total excess above $c_{i,(n+1)}$ is less than $n\delta$.
To make this precise, let $q=\lfloor n/4\rfloor$.
Since each of the first $q+1$ chores has cost at least $c_{i,(q+1)}$, we have
\begin{equation*}
(q+1)\cdot\left(c_{i,(q+1)}-c_{i,(n+1)}\right)
\le \excess{i}<n\cdot\delta.
\end{equation*}
Since $c_{i,(n+1)}\le0.5$ and $q+1>n/4$, we obtain
\begin{equation*}
c_{i,(q+1)}
< c_{i,(n+1)}+\frac{n\cdot\delta}{q+1}
<0.5+4\cdot\delta=0.7.
\end{equation*}
Thus, assigning zero fractions to the first $q$ chores reduces
the largest cost among chores of positive fraction from close
to $1$ to below $0.7$.

These observations explain the two adjustments in
Example~\ref{example:intermediate-allocations}.
In the first intermediate allocation, increasing the fractions on $H_{i,[25]}$ adds little cost because these chores have small total cost. 
Decreasing the fractions on the remaining chores more than offsets this increase, which reduces the fractional cost while leaving the largest eligible cost unchanged.
In the second intermediate allocation, setting the fractions on $H_{i,[25]}$ to zero reduces the largest eligible cost to $c_{i,(26)}$. 
This allows a higher fractional cost while still improving the rounding guarantee.
Thus, the two adjustments improve different terms of the bound in \Cref{lem:sorted-slot-rounding}.

\paragraph{Coordinating Adjustments Across Agents.}
The adjustments above improve an individual hard agent's
rounding guarantee, but increasing her fraction of a chore
requires decreasing someone else's.
To ensure the equal-marginal property, our construction therefore depends on the number of hard agents.
When there are at least two, we arrange them in pairs and,
if needed, one triple, whose members compensate for one another's
changes. We develop the construction for two hard agents (Section~\ref{subsec:hard-pair}) and extend it to more hard agents (Section~\ref{subsec:hard-triple}).
When there is only one hard agent, the easy agents must absorb her adjustments. Section~\ref{subsec:hard-singleton} uses their slack to do so while preserving their $1.97$-MMS guarantee.
Finally, Section~\ref{subsec:proof} combines these cases into the
complete mechanism and proves truthfulness.

\subsection{Two Hard Agents}
\label{subsec:hard-pair}
\label{subsubsec:pair}

Throughout the remaining analysis, let $q:=\left\lfloor n/4\right\rfloor$.
Therefore, $H_{i,[q]}$ contains the top-$\left\lfloor n/4\right\rfloor$ most costly chores of agent $i$.
For convenience, we slightly abuse the notation and use $H_i$ to denote $H_{i,[q]}$.
Consider two hard agents $i,j$.
Since $\left|H_{i}\cup H_{j}\right|\le2q$ and $c_{i,(n+1)}\le0.5$, \Cref{eqn:hard-Ei-average-pi} gives
\begin{equation}\label{eq:union-cost}
c_i\!\left(H_{i}\cup H_{j}\right)\le c_{i,(n+1)}\cdot\left|H_{i}\cup H_{j}\right|+\excess{i}<q+0.05\cdot n\qquad \forall i,j\in Q.
\end{equation}
We remark that this bound does not require the two agents to rank the chores similarly.

\paragraph{The Two Fractional Rules.}
For this pair, we index the two alternatives by $s \in \{1,2\}$, where $s=1$ selects agent $i$ and $s=2$ selects agent $j$, each with probability $1/2$.
The resulting allocation $x^{(1)}$ (corresponding to $s=1$) is defined as the following fractions:
\begin{center}
\begin{tabular}{lcc}
\toprule
Chore & Selected member $i$ & Other member $j$\\
\midrule
$e\in H_{i}\cup H_{j}$ & $2/n$ & $0$\\[2pt]
$e\notin H_{i}\cup H_{j}$ & $1/(2n)$ & $3/(2n)$\\
\bottomrule
\end{tabular}
\end{center}
When agent $j$ is selected ($s=2$), the assignment $x^{(2)}$ is obtained by swapping the columns.
In either intermediate allocation, the pair's total fraction of every chore is $2/n$, and each member's average fraction is $1/n$.
Therefore, if these are the only two hard agents, every other agent keeps fraction $1/n$ of every chore.

By \Cref{eq:union-cost}, the selected member's fractional cost is
\begin{equation*}
\frac{2}{n}\cdot c_i\!\left(H_{i}\cup H_{j}\right)+\frac{1}{2n}\cdot c_i\!\left(M\setminus\left(H_{i}\cup H_{j}\right)\right)<0.5+1.5\cdot\left(\frac{q}{n}+0.05\right).
\end{equation*}
Since every chore has cost at most $1$ under the normalization, \Cref{lem:sorted-slot-rounding} yields
\begin{equation}\label{eq:pair-selected}
c_i\!\left(X_i\right)<1+0.5+1.5\cdot\left(\frac{q}{n}+0.05\right) \leq 1.5 + 1.5\times 0.3 = 1.95.
\end{equation}
On the other hand, suppose hard agent $i$ has fraction zero on $H_i$ and fraction at most $3/(2n)$ on every other chore.
The chores in $H_{i,[q+1]}$ satisfy
\begin{equation*}
c_{i,(1)}-c_{i,(n+1)}+q\cdot\left(c_{i,(q+1)}-c_{i,(n+1)}\right)\le\excess{i}<n\cdot\left(c_{i,(1)}-0.95\right).
\end{equation*}
Rearranging terms and applying $c_{i,(1)}\le1$ and $c_{i,(n+1)}\le0.5$ yields $c_{i,(q+1)}<0.5+\frac{0.05n-0.5}{q}$.
Combining this with $c_i\!\left(H_{i}\right)\ge q\cdot c_{i,(q+1)}$, \Cref{lem:sorted-slot-rounding} implies that
\begin{equation}
\label{eq:nonselected-bound}
c_i\!\left(X_i\right)
\le c_{i,(q+1)}+\frac{3}{2n}\cdot\left(c_i\!\left(M\right)-c_i\!\left(H_{i}\right)\right)
\le1.5+\left(1-\frac{3q}{2n}\right)\cdot c_{i,(q+1)} < 1.95.
\end{equation}
where the last inequality holds for all $n\ge 4$ and $q = \lfloor n/4\rfloor$.
Thus both members of the pair have cost strictly below $1.95$ in every rounded allocation.
When these are the only hard agents, every other agent keeps her uniform fractions and has cost at most $1.95$.

\subsection{More Than Two Hard Agents}
\label{subsec:hard-triple}
\label{subsubsec:triple}
We now extend the pairing construction from \Cref{subsec:hard-pair} to general settings with more than two hard agents.
When the number of hard agents is even, we partition them into pairs and directly apply the construction in \Cref{subsec:hard-pair}.
If the number is odd and greater than $1$, we form exactly one triple and pair all remaining hard agents.

Below, we detail the construction for a triple $\{i, j, k\}$, using the same parameters $q = \lfloor n/4 \rfloor$ and $\delta = 0.05$.
We index the three candidate configurations by $s \in \{1, 2, 3\}$, where $s = 1, 2, 3$ corresponds to selecting agent $i, j,$ and $k$, respectively, with equal probability $1/3$.
Let $x^{(s)}$ denote the intermediate allocation associated with choice $s$.
After redistributing the triple's total fraction $3/n$ for each chore, the final fractions of agents $i,j,k$ are summarized below.

\begin{center}
\begin{tabular}{lccc}
\toprule
Chore & Agent $i$ & Agent $j$ & Agent $k$\\
\midrule
$e\in \left(H_{i}\cup H_{j}\right)\cap\left(H_{i}\cup H_{k}\right)$
& $3/n$ & $0$ & $0$\\[2pt]

$e\in \left(H_{i}\cup H_{j}\right)\setminus\left(H_{i}\cup H_{k}\right)$
& $3/(2n)$ & $0$ & $3/(2n)$\\[2pt]

$e\in \left(H_{i}\cup H_{k}\right)\setminus\left(H_{i}\cup H_{j}\right)$
& $3/(2n)$ & $3/(2n)$ & $0$\\[2pt]

$e\notin \left(H_{i}\cup H_{j}\cup H_{k}\right)$
& $0$ & $3/(2n)$ & $3/(2n)$\\
\bottomrule
\end{tabular}
\end{center}

When agent $i$ is selected, the fractions in the table can be written as follows.
Writing $\mathbf{1}(P)$ for the indicator of a condition $P$ (equal to $1$ when $P$ holds and $0$ otherwise), we have
\begin{equation*}
\begin{aligned}
x_{ie}^{(1)}&=\frac{3}{2n}\cdot\left(\mathbf{1}\left(e\in H_{i}\cup H_{j}\right)+\mathbf{1}\left(e\in H_{i}\cup H_{k}\right)\right),\\
x_{je}^{(1)}&=\frac{3}{2n}\cdot\mathbf{1}\left(e\notin H_{i}\cup H_{j}\right),\qquad
x_{ke}^{(1)}=\frac{3}{2n}\cdot\mathbf{1}\left(e\notin H_{i}\cup H_{k}\right).
\end{aligned}
\end{equation*}
The entries of $x^{(2)}$ and $x^{(3)}$ are obtained by exchanging roles.
In every intermediate allocation, the triple contributes a total fraction of $3/n$ to every chore.

To verify that the average fractions are uniform, fix an agent $i$ in the triple and a chore $e\in M$.
We consider all three choices of the selected agent.
When $i$ is selected, her fraction $x_{ie}^{(1)}$ is given above.
When $j$ or $k$ is selected, respectively, her fraction is
\begin{equation*}
x_{ie}^{(2)}=\frac{3}{2n}\mathbf{1}(e\notin H_i\cup H_j),
\qquad
x_{ie}^{(3)}=\frac{3}{2n}\mathbf{1}(e\notin H_i\cup H_k).
\end{equation*}
Since each choice has probability $1/3$, her average fraction is
\begin{align*}
\frac{x_{ie}^{(1)}+x_{ie}^{(2)}+x_{ie}^{(3)}}{3}
&=\frac{1}{2n}\left(
\mathbf{1}(e\in H_i\cup H_j)+\mathbf{1}(e\notin H_i\cup H_j)
+\mathbf{1}(e\in H_i\cup H_k)+\mathbf{1}(e\notin H_i\cup H_k)
\right)\\
&=\frac{1}{2n}(1+1)=\frac1n.
\end{align*}
The last equality holds for every chore: whether or not it belongs to either union, the corresponding two indicators sum to one.
Thus every agent in the triple receives an average fraction of $1/n$ of every chore.
We next bound the cost of each member after rounding. For the selected member $i$, \Cref{eq:union-cost} gives fractional cost
\begin{equation*}
\frac{3}{2n}\cdot\left(c_i\!\left(H_{i}\cup H_{j}\right)+c_i\!\left(H_{i}\cup H_{k}\right)\right)<3\cdot\left(\frac{q}{n}+0.05\right).
\end{equation*}
Adding at most one additional chore yields
\begin{equation}\label{eq:triple-selected}
c_i\!\left(X_i\right)<1+3\cdot\left(\frac{q}{n}+0.05\right)\le 1.9.
\end{equation}
The nonselected agent $j$ (resp. $k$) has zero fraction on $H_{j}$ (resp. $H_{k}$) and a fraction at most $3/(2n)$ on every other chore.
Hence \Cref{eq:nonselected-bound} applies without change, and every nonselected member has cost strictly below $1.95$.

\paragraph{Combining Pairs and Triples.}
For $|Q|\ge2$, we partition the hard agents into pairs using a fixed public order, with one triple if $|Q|$ is odd.
We apply the constructions above within each group and keep every easy agent's fractions at $1/n$.
For the pairs, select either the first member of every pair or the second member of every pair, each with probability $1/2$.
If a triple is present, independently select one of its members with probability $1/3$ each.
Recall that $R$ denotes the number of intermediate fractional allocations. Thus, $R=2$ for pairs only, $R=3$ for a single triple, and $R=6$ when both are present, with all $R$ allocations equally likely.
Each group preserves its total fraction of every chore, so the combined allocations are feasible.
Their average is the uniform allocation, and rounding preserves these fractions in expectation, giving an equal-marginal mechanism.
The preceding cost bounds ensure that every realized allocation gives each agent a normalized cost of at most $1.95$.
\subsection{Exactly One Hard Agent}
\label{subsec:hard-singleton}
It remains to consider the most challenging case of $|Q|=1$.
Let agent $1$ be the unique hard agent.
Every other agent $j$ is easy and therefore satisfies $\costbound{j}\le1.95$, leaving a slack of $0.02$ before the target $1.97$.
Note that this is the \emph{only} part of the proof where we use the assumption $n\ge19$.
The reason is structural: with a unique hard agent, no other hard agent is available to offset the adjustment as in previous constructions.
We therefore spread the adjustment of agent $1$ among all $n-1$ easy agents.
The lower bound $n\ge19$ is used to ensure that the additional cost absorbed by each easy helper is smaller than her available slack of $0.02$.

\paragraph{The Two Intermediate Fractional Allocations.}
Let $h$ be the number of chores whose cost to agent $1$ is strictly larger than $0.94$.
By the fixed tie-breaking order, these chores are exactly $H_{1,[h]}$.
Let $\Delta:=\frac{1}{n\cdot\left(n-1\right)}$.
We construct two intermediate fractional allocations $x^{(1)}$ and $x^{(2)}$, each chosen with probability $0.5$.
Starting from the uniform allocation $x_{ie} = 1/n$ for every agent $i$ and chore $e$, we exchange fractions between agent $1$ and each easy agent $j \neq 1$ according to the following table, where the values of $a_j$ and $b$ will be determined later.
%
%
The choice of $\Delta$ ensures that in allocation $x^{(2)}$, agent $1$ receives $0$ fraction of chores in $H_{1,[h]}$, after exchanging with the other $n-1$ agents.

\begin{center}
\begingroup
\renewcommand{\arraystretch}{1.25}
\begin{tabular}{lcc}
\toprule
 & \multicolumn{2}{c}{Change in agent $1$'s fraction from the exchange with $j$}\\
\cmidrule(lr){2-3}
Chore & Intermediate allocation $1$ & Intermediate allocation $2$\\
\midrule
$e\in H_{1,[h]}$ & $+\Delta$ & $-\Delta$\\
$e\in H_{j,[n]}\setminus H_{1,[h]}$ & $-a_j\cdot\Delta$ & $+a_j\cdot\Delta$\\
$e\notin H_{j,[n]}\cup H_{1,[h]}$ & $-b\cdot\Delta$ & $+b\cdot\Delta$\\
\bottomrule
\end{tabular}
\endgroup
\end{center}

\paragraph{Choice of Parameters $a_j$ and $b$.}
Recall that for an easy agent $j$, the strengthened bound of \Cref{lem:sorted-slot-rounding} ensures that the rounded cost is at most $\Gamma_j \leq 1.95$.
However, the strengthened bound is applicable only when $\sum_{e\in H_{j,[n]}} x_{j e} = 1$.
Therefore, we need to ensure that the exchange of fractions between agents $1$ and $j$ preserves this key property.
This can be ensured by setting
\begin{equation*}
    a_j:=\frac{\left|H_{1,[h]}\cap H_{j,[n]}\right|}{n-\left|H_{1,[h]}\cap H_{j,[n]}\right|} = \frac{\left|H_{1,[h]}\cap H_{j,[n]}\right|}{\left| H_{j,[n]} \setminus H_{1,[h]} \right|}.
\end{equation*}
Then the total change of agent $j$'s fraction on $H_{j,[n]}$ is zero in both intermediate allocations.
%

Next, we select a uniform value of $b$ such that the combined fraction exchanges reduce agent $1$'s fractional cost in the first allocation by at least $0.03$.
Specifically, we set
\begin{equation*}
b:=\frac{\sum_{j\ne1}\left(c_1\!\left(H_{1,[h]}\right)-a_j\cdot c_1\!\left(H_{j,[n]}\setminus H_{1,[h]}\right)+0.03\cdot n\right)}{\sum_{j\ne1}c_1\!\left(M\setminus\left(H_{j,[n]}\cup H_{1,[h]}\right)\right)}.
\end{equation*}
\paragraph{Validity of the Construction.}
By \Cref{eqn:hard-Ei-average-pi}, we have
\begin{equation*}
\excess{1}<0.05\cdot n,\qquad \frac{c_1\!\left(M\right)}{n}>0.95+\frac{\excess{1}}{n},\qquad c_{1,(1)}>0.95.
\end{equation*}
In particular, $H_{1,[h]}$ is nonempty.
Every chore in $H_{1,[h]}$ contributes more than $0.94-0.5=0.44$ to $\excess{1}$.
Therefore, we have
\begin{equation*}
|H_{1,[h]}| < \frac{\excess{1}}{0.44} < \frac{5n}{44},\qquad c_1\!\left(H_{1,[h]}\right)<\frac{47}{440}\cdot n.
\end{equation*}
Consequently, $\left|H_{1,[h]}\cap H_{j,[n]}\right|<\frac{5n}{44}$ and $0\le a_j<\frac{5}{39}$.
By the definition of $\excess{1}$, we have
\begin{equation*}
c_1\!\left(H_{j,[n]}\cup H_{1,[h]}\right)\leq c_{1,(n+1)}\cdot \left|H_{1,[h]}\cup H_{j,[n]}\right| + E_1 
< 0.5\cdot \left( n+\frac{\excess{1}}{0.44} \right) + \excess{1}
= \frac{n}{2}+\frac{47\excess{1}}{22},
\end{equation*}
which implies that 
\begin{equation*}
\frac{c_1\!\left(M\setminus\left(H_{j,[n]}\cup H_{1,[h]}\right)\right)}{n}>0.45-\frac{25}{22}\cdot\frac{\excess{1}}{n}>\frac{173}{440}.
\end{equation*}
The denominator in the definition of $b$ is therefore positive.
Moreover, for every $j\ne1$, every chore in $H_{j,[n]}\setminus H_{1,[h]}$ costs at most $0.94$ to agent $1$, while every chore in $H_{1,[h]}$ costs more than $0.94$. Hence,
\begin{equation*}
a_j\cdot c_1\!\left(H_{j,[n]}\setminus H_{1,[h]}\right)
\le 0.94\cdot\left|H_{1,[h]}\cap H_{j,[n]}\right|
<c_1\!\left(H_{1,[h]}\right),
\end{equation*}
so the numerator is positive as well and $b>0$.
By definition of $b$ and the above inequalities, we obtain
\begin{equation*}
b < \frac{\sum_{j\ne1} \left( c_1\!\left(H_{1,[h]}\right) + 0.03\cdot n \right)}{\frac{173}{440}\cdot n\cdot (n-1)}
< \frac{47/440+0.03}{173/440}=\frac{301}{865}<0.35.
\end{equation*}

Therefore, all easy agents' fractions (even after being decreased by $b\cdot \Delta$) are positive.
On every chore in $H_{1,[h]}$, agent $1$'s two fractions are $2/n$ and $0$.
Outside $H_{1,[h]}$, her fraction in allocation $1$ is $1/n$ minus $n-1$ nonnegative transfers, each strictly smaller than $\Delta$, so it lies in $\left(0,1/n\right]$.
Her fraction in allocation $2$ is its complement to $2/n$, so it lies in $\left[1/n,2/n\right)$.
Thus both intermediate allocations are feasible.

It remains to bound the costs after rounding, starting with the hard agent. By the choice of $b$, the exchanges together change agent $1$'s fractional cost in allocation $1$ by $-0.03$.
Reversing the changes in allocation $2$ yields
\begin{equation*}
\sum_e c_1\!\left(e\right)\cdot x^{(1)}_{1e}=\frac{c_1\!\left(M\right)}{n}-0.03,\qquad \sum_e c_1\!\left(e\right)\cdot x^{(2)}_{1e}=\frac{c_1\!\left(M\right)}{n}+0.03.
\end{equation*}
Recall that the largest eligible chore costs at most $1$ in allocation $1$ and at most $0.94$ in allocation $2$.
Therefore, we have $c_1\!\left(X_1\right)\le 1.97$ after rounding in both intermediate allocations.

For each easy agent $j$, the total fraction of $H_{j,[n]}$ remains one, so the strengthened bound of \Cref{lem:sorted-slot-rounding} still applies, and the increase in the bound comes from the change in fractional cost.
In allocation $1$, only chores outside $H_{j,[n]}\cup H_{1,[h]}$ can increase the fractional cost, and their contribution is at most
\begin{equation*}
b\cdot\Delta\cdot c_j\left(M\setminus\left(H_{j,[n]}\cup H_{1,[h]}\right)\right)\le\frac{b}{n-1}<\frac{0.35}{n-1}\le\frac{0.35}{18}<0.02.
\end{equation*}
Here we use $n\ge19$. 
In allocation $2$, only chores in $H_{1,[h]}\setminus H_{j,[n]}$ can increase the fractional cost.
Each costs at most $0.5$ to agent $j$, so the increase is at most
\begin{equation*}
\Delta\cdot c_j\left(H_{1,[h]}\setminus H_{j,[n]}\right)\le\frac{|H_{1,[h]}|}{2n\cdot\left(n-1\right)}<\frac{5}{88\cdot\left(n-1\right)}\le\frac{5}{88\times 18}<0.02.
\end{equation*}
Thus every easy agent satisfies $c_j\left(X_j\right)<\costbound{j}+0.02\le1.95+0.02=1.97$.
This completes the one-hard-agent case.

\subsection{Complete Mechanism and Truthfulness}
\label{subsec:proof}

Finally, we combine the preceding constructions to obtain the complete mechanism.
Given the reported costs, we compute $\costbound{i}$ for each agent and identify the set $Q$ of hard agents.\footnote{If $|M|\le n$, we first add enough zero-cost dummy chores to ensure $|M|>n$.}
If there are no hard agents, we round the uniform fractional allocation directly.
Otherwise, we use the intermediate allocations from Section~\ref{subsec:hard-singleton} when $|Q|=1$, or those from Section~\ref{subsec:hard-triple} when $|Q|\ge2$.
We sample one of these allocations uniformly and round it using \Cref{lem:sorted-slot-rounding}.

The preceding analysis shows that every resulting allocation gives each agent a normalized cost of at most $1.97$.
Since $L_i\le\mu_i$, this proves the ex-post $1.97$-MMS guarantee.
Moreover, the intermediate allocations average to the uniform allocation, and rounding preserves their fractions in expectation.
Thus, each agent receives each chore with probability $1/n$ at every reported profile.
Her expected cost is therefore $c_i(M)/n$ regardless of her report, which establishes TIE and ex-ante EF.
Hence, together with the complexity analysis below, we obtain the following.

\begin{theorem}\label{thm:1.97}
For every $n\ge19$, there is a polynomial-time randomized equal-marginal mechanism for additive chores that is TIE, ex-ante EF and ex-post
$1.97$-MMS.
\end{theorem}
Combining this mechanism with the mechanisms in \Cref{sec:nomination_overview} for small $n$ yields a uniform approximation ratio strictly below $2$ for all $n$.

\paragraph{Computational Complexity.}
The mechanism can be implemented in time polynomial in the input bit size. Sorting the chores and computing $L_i$, $\costbound{i}$, the hard set, and all intermediate fractional allocations require polynomially many arithmetic operations. For any sampled intermediate allocation, the sorted-slot construction creates at most $m+n$ slots. The resulting fractional perfect matching can be decomposed in polynomial time by the standard algorithm for the Birkhoff--von Neumann decomposition, which repeatedly finds a perfect matching in the positive support and subtracts the minimum weight on that matching. Sampling a matching from the resulting decomposition then produces the final integral allocation.

\section{Lower Bounds}
\label{sec:lower_bounds}
In this section, we establish a lower bound on the MMS approximation ratio for TIE mechanisms. 
Our primary negative result is formalized below.

\begin{theorem}\label{thm:tie-lower-bounds}
For every integer $n\geq 2$, no TIE mechanism can guarantee an ex-post MMS approximation ratio strictly smaller than $\alpha_n$, where
\begin{equation*}
\alpha_n =
\begin{cases}
\dfrac{4}{3},&n=2,\\[4pt] 
\max\left\{1+\dfrac{1}{3\cdot n-2}, \frac{13}{12} \right\},&n\geq 3\text{ odd},\\[4pt] 
\max\left\{ 1+\dfrac{1}{3\cdot n-1}, \frac{13}{12} \right\},&n\geq 4\text{ even}.
\end{cases}
\end{equation*}
\end{theorem}

We prove the theorem by presenting two sets of hard instances.
The first one establishes a lower bound of $1+\Theta(1/n)$ on the approximation ratio of TIE mechanisms for every $n\geq 2$.
While the ratio approaches $1$ when $n$ is large, it is meaningful when $n$ is small.
For example, when $n = 2$, the lower bound $4/3$ matches the upper bound given by our TIE mechanism for two agents (see Appendix~\ref{sec:nomination-two}), and is therefore tight.
The second one establishes a lower bound of $13/12$ for every $n\geq 4$ using a two-report construction and a single-agent deviation argument.

\subsection{TIE Mechanisms: Two Agents}

We first present the lower bound of $4/3$ for two agents as a warm-up.
Similar proof ideas will be reused in the later analysis for $n\geq 3$ agents.

\begin{theorem}
\label{thm:two}
For two agents, no complete TIE mechanism can guarantee ex-post $\alpha$-MMS for every additive chore instance when $\alpha<4/3$.
\end{theorem}

There are four chores $e_1,e_2,e_3,e_4$, where $e_3$ and $e_4$ are the unit-cost chores.
Fix $\alpha<4/3$, and choose $\varepsilon>0$ sufficiently small that
\begin{equation*}
\alpha\cdot\left(3+\varepsilon\right)<4 \qquad\text{and}\qquad \alpha\cdot\left(2+\varepsilon\right)<3.
\end{equation*}
We use the three reports in Table~\ref{tab:warm-reports}. Here $L$ and $H$ indicate low and high cost for $e_1$, while $A$ denotes the auxiliary report.
We use $c_L$, $c_A$, and $c_H$ to denote the corresponding cost functions.

\begin{table}[ht]
\centering
\renewcommand{\arraystretch}{1.2}
\begin{tabular}{c c c c c c}
\toprule
Report & $e_1$ & $e_2$ & $e_3$ & $e_4$ & MMS \\
\midrule
$L$ & $\varepsilon$ & $2$ & $1$ & $1$ & $2+\varepsilon$ \\
$A$ & $2+\varepsilon$ & $2$ & $1$ & $1$ & $3+\varepsilon$ \\
$H$ & $3$ & $1$ & $1$ & $1$ & $3$ \\
\bottomrule
\end{tabular}
\caption{The three reports for the two-agent lower bound.}
\label{tab:warm-reports}
\end{table}

By the choice of $\varepsilon$, the ex-post guarantee implies that every supported $L$-bundle has reported cost strictly below $3$, while every supported $A$- or $H$-bundle has reported cost strictly below $4$.

These strict caps immediately give the bundle restrictions that drive the proof:

\begin{itemize}
\item an $H$-agent receiving $e_1$ receives nothing else;
\item an $L$-agent receiving $e_2$ receives neither unit-cost chore;
\item an $A$-agent receiving $e_1$ receives no $e_2$ and at most one unit-cost chore;
\item an $A$-agent receiving $e_2$ receives at most one unit-cost chore.
\end{itemize}

Fix agent $1$. We use the four profiles $(L,H)$, $(L,A)$, $(A,H)$, and $(H,H)$. Their roles are summarized below.
Figure~\ref{fig:warm-chain} shows the incentive chain, e.g., the arrow from $(L,H)$ indicates that agent $1$ with true cost $L$ should not have an incentive to deviate from $(L,H)$ to $(A,H)$. The auxiliary profile $(L,A)$ is used only in the first step.

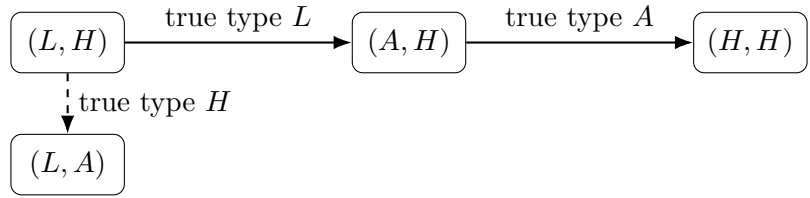
\begin{figure}[ht]
\centering
\begin{tikzpicture}[
node distance=8mm and 30mm,
box/.style={draw,rounded corners,align=center,minimum width=15mm,minimum height=8mm},
arr/.style={-{Latex[length=2.4mm]},thick}
]
\node[box] (LH) {$(L,H)$};
\node[box, right=of LH] (AH) {$(A,H)$};
\node[box, right=of AH] (HH) {$(H,H)$};
\node[box, below=of LH] (LA) {$(L,A)$};
\draw[arr] (LH) -- node[above,align=center]{true type $L$} (AH);
\draw[arr] (AH) -- node[above,align=center]{true type $A$} (HH);
\draw[arr,dashed] (LH) -- node[right,align=left]{true type $H$} (LA);
\end{tikzpicture}
\caption{The incentive comparisons used in the two-agent warm-up.}
\label{fig:warm-chain}
\end{figure}

Our analysis consists of three main steps.
In the first step, we focus on the profile $(L,H)$.
We show that agent $1$ has a high cost at $(L,H)$, because agent $2$'s cost must be small, otherwise she has an incentive to misreport $A$.
Then we show in the second step that agent $1$ must have no smaller cost at $(A,H)$, because otherwise she has an incentive to misreport $A$ in $(L,H)$.
In the third step, we show that agent $1$'s expected $c_A$-cost at $(H,H)$ must be at least $3+\varepsilon$. Otherwise, with true cost $c_A$, she could reduce her expected cost by reporting $H$ instead of $A$ when agent $2$ reports $H$.
Repeating the argument with the agents' roles exchanged gives the same lower bound for agent $2$. Their total expected $c_A$-cost would therefore be at least $6+2\varepsilon$, contradicting $c_A(M)=6+\varepsilon$.

\paragraph{Step 1.}
Consider the auxiliary profile $(L,A)$. Agent $2$ cannot receive $e_1$ at this profile because otherwise agent $1$ would have to receive $e_2$ and one unit-cost chore, which leads to an approximation ratio larger than $\alpha$.
Given that agent $2$ does not receive $e_1$ at $(L,A)$, she can receive at most two chores in $\left\{e_2,e_3,e_4\right\}$, which implies that her bundle has $H$-cost at most $2$.

Now regard $H$ as agent $2$'s true cost. Truthful reporting produces $(L,H)$, while deviating to $A$ produces $(L,A)$. TIE gives $\E\left[c_H\left(X_2\left(L,H\right)\right)\right]\leq 2$.
Since the total $H$-cost of all four chores is $6$, we further have $\E\left[c_H\left(X_1\left(L,H\right)\right)\right]\geq 4$.
It can be verified that for every bundle $Y$, we have $c_L\left(Y\right)\geq c_H\left(Y\right)-\left(3-\varepsilon\right)$.
Taking expectations gives
\begin{equation*}
\E\left[c_L\left(X_1\left(L,H\right)\right)\right]\geq 1+\varepsilon.
\end{equation*}

\paragraph{Step 2.}
At $(A,H)$, chore $e_1$ must be assigned to agent $1$, because otherwise the approximation ratio will be strictly larger than $\alpha$.
Then agent $1$ cannot receive $e_2$, and can receive at most one unit-cost chore.
Regard $L$ as her true cost. Truthful reporting produces $(L,H)$, while reporting $A$ produces $(A,H)$. By TIE, we obtain
\begin{equation*}
\E\left[c_L\left(X_1\left(A,H\right)\right)\right] \geq \E\left[c_L\left(X_1\left(L,H\right)\right)\right] \geq 1+\varepsilon.
\end{equation*}
At $(A,H)$, agent $1$ receives $e_1$, no $e_2$, and at most one unit-cost chore. Hence the preceding lower bound forces her to receive one unit-cost chore, which implies $\E\left[c_A\left(X_1\left(A,H\right)\right)\right]\geq 3+\varepsilon$.

\paragraph{Step 3.}
Now regard $A$ as agent $1$'s true cost.
TIE gives
\begin{equation*}
\E\left[c_A\left(X_1\left(H,H\right)\right)\right] \geq \E\left[c_A\left(X_1\left(A,H\right)\right)\right] \geq 3+\varepsilon.
\end{equation*}
Repeating Steps 1 and 2 with the roles of agents $1$ and $2$ exchanged gives $\E\left[c_A\left(X_2\left(H,A\right)\right)\right]\geq 3+\varepsilon$.
Now regard $A$ as agent $2$'s true cost. Truthful reporting gives $(H,A)$, while deviating to $H$ gives $(H,H)$. Therefore,
\begin{equation*}
\E\left[c_A\left(X_2\left(H,H\right)\right)\right] \geq \E\left[c_A\left(X_2\left(H,A\right)\right)\right] \geq 3+\varepsilon.
\end{equation*}
However, we have $c_A\left(M\right)=6+\varepsilon$, which is a contradiction.

\subsection{Generalizing to Three or More Agents}

We now extend the same three-step argument to every $n\geq 3$.
Fix $s=\left\lceil 3\left(n-1\right)/2\right\rceil \geq n$.
There is one distinguished chore $e$, there are $n-1$ large chores, and there are $\left(n-1\right)\cdot s$ small chores.
Fix $\alpha<1+1/\left(2s+1\right)$, and choose $\varepsilon>0$ sufficiently small that
\begin{equation}
\alpha\cdot\left(2s+\varepsilon\right)<2 s+1 \qquad\text{and}\qquad \alpha\cdot\left(\left(n-1\right) s+\varepsilon\right)<\left(n-1\right) s+s-n+1. \label{eq:general-epsilon}
\end{equation}
Such a choice is possible from the definition of $s$ and the strict inequality $\alpha<1+1/\left(2 s+1\right)$.
We use the three reports in Table~\ref{tab:general-reports}, and write $c_L$, $c_A$, and $c_H$ for the corresponding cost functions.

\begin{table}[ht]
\centering
\small
\renewcommand{\arraystretch}{1.2}
\begin{tabular}{c >{\centering\arraybackslash}p{27mm} >{\centering\arraybackslash}p{33mm} >{\centering\arraybackslash}p{27mm} >{\centering\arraybackslash}p{34mm}}
\toprule
Report & $e$ & large chore & small chore & MMS \\
\midrule
$L$ & $\varepsilon$ & $\left(n-1\right) s$ & $1$ & $\left(n-1\right) s+\varepsilon$ \\
$A$ & $2 s+1-n+\varepsilon$ & $s+1$ & $1$ & $2 s+\varepsilon$ \\
$H$ & $2 s+1$ & $s+1$ & $1$ & $2 s+1$ \\
\bottomrule
\end{tabular}
\caption{The three reports for the lower bound with $n\geq 3$.}
\label{tab:general-reports}
\end{table}

Note that the MMS in the last column can be verified directly. 
To establish the upper bounds, we construct the bundles as follows. 
For $L$, we place each large chore alone and combine $e$ with all small chores. 
For $A$, we pair each large chore with $s-1$ small chores and group $e$ with the remaining $n-1$ small chores. 
Finally, for $H$, we place $e$ alone and bundle each large chore with $s$ small chores.
By construction, we have
\begin{equation}
c_H\left(M\right)=n\cdot\left(2s+1\right), \qquad c_A\left(M\right)=2 n\cdot s+\varepsilon. \label{eq:general-totals}
\end{equation}

By \Cref{eq:general-epsilon}, every supported $H$-bundle has cost strictly below $2 s+2$, every supported $A$-bundle has cost strictly below $2 s+1$, and every supported $L$-bundle has cost strictly below $\left(n-1\right) s+s-n+1$.
As in the two-agent warm-up, these caps immediately give the bundle restrictions needed below:
\begin{itemize}
\item no agent receives two large chores;
\item an $H$-agent receiving $e$ receives no other chore;
\item an $L$-, $A$-, or $H$-agent receiving a large chore but not $e$ receives at most $s-n$, $s-1$, or $s$ small chores, respectively; otherwise, the $\alpha$-MMS guarantee will be violated (see \Cref{eq:general-epsilon});
\item an $A$-agent receiving $e$ receives no large chore and at most $n-1$ small chores;
\item an $A$-agent not receiving $e$ has $H$-cost at most $2s$.
\end{itemize}

Fix an arbitrary agent $i$.
The proof follows the same three-step chain.

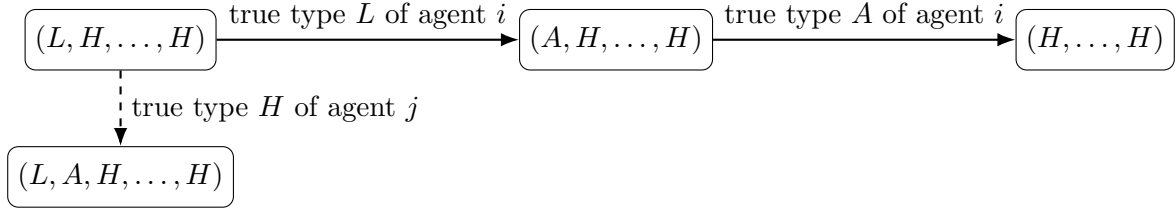
\begin{figure}[ht]
\centering
\begin{tikzpicture}[
node distance=10mm and 40mm,
box/.style={draw,rounded corners,align=center,minimum width=15mm,minimum height=8mm},
arr/.style={-{Latex[length=2.4mm]},thick}
]
\node[box] (Pi) {$(L,H,\ldots,H)$};
\node[box, right=of Pi] (Qi) {$(A,H,\ldots,H)$};
\node[box, right=of Qi] (R) {$(H,\ldots,H)$};
\node[box, below=of Pi] (Dij) {$(L,A,H,\ldots,H)$};

\draw[arr] (Pi) -- node[above,align=center]{true type $L$ of agent $i$} (Qi);
\draw[arr] (Qi) -- node[above,align=center]{true type $A$ of agent $i$} (R);
\draw[arr,dashed] (Pi) -- node[right,align=left]{true type $H$ of agent $j$} (Dij);
\end{tikzpicture}
\caption{The incentive comparisons used for the general construction.}
\label{fig:general-chain}
\end{figure}

\paragraph{Step 1.}
Let $P_i$ be the profile in which agent $i$ reports $L$ and every other agent reports $H$.
We claim that $\E\left[c_L\left(X_i\left(P_i\right)\right)\right]\geq n-1+\varepsilon$.
Fix an opponent $j\neq i$, and let $D_{ij}$ be obtained from $P_i$ by changing only agent $j$ from $H$ to $A$.
Agent $j$ cannot receive $e$ at $D_{ij}$.
Indeed, if she did, she could receive no large chore and at most $n-1$ small chores, so the $n-1$ large chores would have to go one each to the remaining agents.
Agent $i$ (the unique $L$-agent) could then take at most $s-n$ small chores with her large chore, while each of the other $n-2$ $H$-agents could take at most $s$ small chores.
Together with agent $j$, the total capacity would be only $\left(n-1\right)\cdot s-1$ small chores, one fewer than the instance contains.
Hence agent $j$ does not receive $e$ at $D_{ij}$, and her bundle has $H$-cost at most $2s$.
Now regard $H$ as agent $j$'s true cost.
Truthful reporting gives $P_i$, while reporting $A$ gives $D_{ij}$, so TIE implies $\E\left[c_H\left(X_j\left(P_i\right)\right)\right]\leq 2s$.
This holds for every $j\neq i$.
Using completeness and \Cref{eq:general-totals},
\begin{equation*}
\E\left[c_H\left(X_i\left(P_i\right)\right)\right]\geq n\cdot\left(2s+1\right)-\left(n-1\right)\cdot 2s=2s+n.
\end{equation*}
For every bundle $Y$, we have $c_L\left(Y\right)\geq c_H\left(Y\right)-\left(2s+1-\varepsilon\right)$, because only $e$ is cheaper under $L$ than under $H$.
Taking expectations proves the claim.

\paragraph{Step 2.}
Let $Q_i$ be the profile in which agent $i$ reports $A$ and every other agent reports $H$.
At $Q_i$, agent $i$ must receive $e$.
Otherwise, some $H$-agent receives $e$ and therefore receives nothing else.
The remaining $n-1$ agents must then receive the $n-1$ large chores, with one chore assigned to each agent. The $A$-agent can take at most $s-1$ small chores with her large chore, and each of the other $n-2$ agents of type $H$ can take at most $s$ small chores.
Their total capacity is again only $\left(n-1\right)\cdot s-1$, a contradiction.
Thus, agent $i$ receives $e$, no large chores, and a subset of small chores. 
For any such bundle $Y$, we have $c_A(Y) - c_L(Y) = 2s + 1 - n$. 
Assuming agent $i$'s true cost function is $L$, truthful reporting yields profile $P_i$, whereas reporting $A$ yields $Q_i$. By TIE, we have
\begin{equation*}
\E\left[c_A\left(X_i\left(Q_i\right)\right)\right] = \E\left[c_L\left(X_i\left(Q_i\right)\right)\right] + (2s+1-n) \geq \E\left[c_L\left(X_i\left(P_i\right)\right)\right] + (2s+1-n) \geq 2s+\varepsilon.
\end{equation*}

\paragraph{Step 3.}
Let $(H,\cdots,H)$ be the profile in which every agent reports $H$.
Fix agent $i$ and regard $A$ as her true cost.
Truthful reporting gives $Q_i$, while reporting $H$ gives $(H,\cdots,H)$.
By TIE, we obtain $\E\left[c_A\left(X_i\left(H,\cdots,H\right)\right)\right]\geq 2s+\varepsilon$.
The same argument applies to every agent, so
\begin{equation*}
\sum_{i=1}^{n}\E\left[c_A\left(X_i\left(H,\cdots,H\right)\right)\right]\geq n\cdot\left(2s+\varepsilon\right),
\end{equation*}
which leads to a contradiction as $c_A\left(M\right)=2\cdot n\cdot s+\varepsilon$.

\medskip

Notice that $2s+1=3\cdot n-2$ when $n$ is odd and $2s+1=3\cdot n-1$ when $n$ is even.
Therefore, the lower bound is $1+1/\left(3\cdot n-2\right)$ for odd $n$ and $1+1/\left(3\cdot n-1\right)$ for even $n$.

\subsection{TIE Mechanisms: Large \texorpdfstring{$n$}{}}

We next give a constant lower bound that is stronger when the number of agents is large.

\begin{theorem}
\label{thm:13-12}
For every $n\geq 4$, no complete TIE mechanism can guarantee an ex-post $\alpha$-MMS allocation for every additive chore instance when $\alpha<13/12$.
\end{theorem}

Fix $n\geq 4$ and $\alpha<13/12$.
Choose an arbitrarily small $\varepsilon>0$ such that $\alpha\cdot\left(11+\varepsilon\right)<12$.

There are eleven core chores: two small chores $s_1,s_2$, four medium chores $t_1,t_2,t_3,t_4$, two chores $u_1,u_2$, two large chores $\ell_1,\ell_2$, and one distinguished chore $e$.
Let $E$ denote the set of all eleven core chores.
In addition, there are $n-4$ huge chores.
The two additive reports are given in Table~\ref{tab:13-12-reports}.
The exaggerated report $c_1$ is identical to $c_0$ on most chores but reports a higher cost on chores $u_1,u_2,\ell_1,\ell_2$.

\begin{table}[ht]
\centering
\setlength{\tabcolsep}{6pt}
\renewcommand{\arraystretch}{1.2}
\begin{tabular}{c c c c c c c c c c c c c c}
\toprule
Report & $s_1$ & $s_2$ & $t_1$ & $t_2$ & $t_3$ & $t_4$ & $u_1$ & $u_2$ & $\ell_1$ & $\ell_2$ & $e$ & \shortstack{each huge chore} & MMS \\
\midrule
$c_0$ & $2$ & $2$ & $3$ & $3$ & $3$ & $3$ & $4+\varepsilon$ & $4$ & $6$ & $6$ & $8$ & $11$ & $11+\varepsilon$ \\
$c_1$ & $2$ & $2$ & $3$ & $3$ & $3$ & $3$ & $5$ & $5$ & $7$ & $7$ & $8$ & $11$ & $12$ \\
\bottomrule
\end{tabular}
\caption{The ordinary report $c_0$ and the exaggerated report $c_1$.}
\label{tab:13-12-reports}
\end{table}

The $n$-agent MMS values of $c_0$ and $c_1$ are exactly $11+\varepsilon$ and $12$, respectively.
For the upper bound for $c_0$, place each huge chore in a singleton bundle and partition the core as
\begin{equation*}
\left\{e,t_1\right\},\qquad \left\{\ell_1,s_1,t_2\right\},\qquad \left\{\ell_2,s_2,t_3\right\},\qquad \left\{u_1,u_2,t_4\right\}.
\end{equation*}
The four core bundles have costs $11,11,11,11+\varepsilon$ under $c_0$.
It can also be verified that for $c_1$, placing each huge chore in a singleton bundle and partitioning the core as
\begin{equation*}
\left\{t_1,t_2,t_3,t_4\right\},\qquad \left\{s_1,s_2,e\right\},\qquad \left\{\ell_1,u_1\right\},\qquad \left\{\ell_2,u_2\right\}
\end{equation*}
gives an MMS partition, where all four core bundles have cost $12$ under $c_1$.

Under the assumed approximation ratio, every supported $c_0$-bundle has cost strictly below $12$, while every supported $c_1$-bundle has integer cost strictly below $13$ and therefore cost at most $12$.
A huge chore cannot share a bundle with any other chore under either report, because every core chore costs at least $2$ and every huge chore costs $11$.
Thus every huge chore forms a singleton bundle in every supported allocation.

\begin{lemma}
\label{lem:13-12-deviation}
Suppose one agent reports $c_1$ and every other agent reports $c_0$.
Then every bundle $Y$ assigned to the $c_1$-reporter in the support of the mechanism satisfies $c_0\left(Y\right)\leq 11$.
\end{lemma}

We first use Lemma~\ref{lem:13-12-deviation} to prove Theorem~\ref{thm:13-12}.
Consider the all-$c_0$ profile and let $X_i$ be the random bundle assigned to agent $i$.
For every named agent $i$, let $Y_i$ denote her random bundle when she alone changes her report from $c_0$ to $c_1$ while every other agent continues to report $c_0$.
Lemma~\ref{lem:13-12-deviation} gives $c_0\left(Y_i\right)\leq 11$ in every supported allocation at the deviation profile.
If the true cost of agent $i$ is $c_0$, TIE therefore implies $\E\left[c_0\left(X_i\right)\right] \leq \E\left[c_0\left(Y_i\right)\right] \leq 11$.
This holds for every agent.
Summing the inequalities gives $c_0(M)\leq 11n$, which is a contradiction as $c_0(M) = 11n + \varepsilon$.
This proves Theorem~\ref{thm:13-12}.

\begin{proofof}{Lemma~\ref{lem:13-12-deviation}}
If the $c_1$-reporter receives a huge chore, then it forms a singleton bundle and has $c_0$-cost exactly $11$.
Suppose therefore that she receives only core chores, and let $Y\subseteq E$ be her bundle.
All $n-4$ huge chores are assigned to other agents, so the remaining core chores $E\setminus Y$ must be allocated among at most three $c_0$-agents.
Each such agent has $c_0$-cost strictly below $12$.
Also, the $c_1$-reporter satisfies $c_1(Y)\leq 12$.

We first identify all bundles with $c_1(Y)\leq 12$ but $c_0(Y)>11$.
Since $c_0\leq c_1$ pointwise, any bundle containing $u_2$, $\ell_1$, or $\ell_2$ satisfies $c_0\left(Y\right)\leq c_1\!\left(Y\right)-1\leq 11$.
Thus, the bundle $Y$ cannot contain any of these chores.
\begin{itemize}
    \item If $Y$ contains $u_1$, then $c_0\left(Y\right)=c_1\!\left(Y\right)-\left(1-\varepsilon\right)$. Because $c_1(Y)$ is an integer at most $12$ and $\varepsilon<1$, the inequality $c_0(Y)>11$ forces $c_1(Y)=12$. Besides $u_1$, the remaining chores in $Y$ must have total $c_1$-cost $7$, and the only way to obtain $7$ from the available costs $2$ and $3$ is $2+2+3$. Hence we have $Y=\left\{s_1,s_2,u_1,t_j\right\}$ for some $j\in\left\{1,2,3,4\right\}$.

    \item If $Y$ does not contain $u_1$, then the two reports agree on every chore in $Y$. Thus, $c_0(Y)>11$ and $c_1(Y)\leq 12$ imply $c_0(Y)=c_1(Y)=12$. If $Y$ contains $e$, the remaining cost is $4$, forcing $Y=\left\{s_1,s_2,e\right\}$. If $Y$ does not contain $e$, the only way to obtain total cost $12$ is to take all four medium chores, so $Y=\left\{t_1,t_2,t_3,t_4\right\}$.
\end{itemize}

Table~\ref{tab:13-12-offending} summarizes the three possible forms.
The fourth column lists the $c_0$-costs of the remaining core chores.
In the last column, we use $\operatorname{MMS}\left(E\setminus Y,3\right)$ to denote the MMS value when the remaining chores are partitioned into three bundles under $c_0$.

\begin{table}[ht]
\centering
\setlength{\tabcolsep}{6pt}
\renewcommand{\arraystretch}{1.25}
\begin{tabular}{>{\centering\arraybackslash}p{32mm} c c >{\centering\arraybackslash}p{46mm} c}
\toprule
$Y$ & $c_0(Y)$ & $c_1(Y)$ & $c_0$-costs in $E\setminus Y$ & $\operatorname{MMS}\left(E\setminus Y,3\right)$ \\
\midrule
$\left\{t_1,t_2,t_3,t_4\right\}$ & $12$ & $12$ & $8,6,6,4+\varepsilon,4,2,2$ & $12$ \\
$\left\{s_1,s_2,e\right\}$ & $12$ & $12$ & $6,6,4+\varepsilon,4,3,3,3,3$ & $12$ \\
$\left\{s_1,s_2,u_1,t_j\right\}$ & $11+\varepsilon$ & $12$ & $8,6,6,4,3,3,3$ & $12$ \\
\bottomrule
\end{tabular}
\caption{The only possible exaggerated bundles with ordinary cost larger than $11$. In every row, the remaining core chores have three-bundle MMS exactly $12$ under $c_0$.}
\label{tab:13-12-offending}
\end{table}

The equality $\operatorname{MMS}\left(E\setminus Y,3\right)=12$ under $c_0$ in Table~\ref{tab:13-12-offending} is easy to verify directly from the displayed costs. In particular, in the first and third rows, any three-partition with maximum cost strictly below $12$ would have to place the chores of costs $8,6,6$ in three different bundles; the second row is checked similarly. Therefore, in every case, $E\setminus Y$ cannot be allocated among at most three $c_0$-agents while keeping every bundle strictly below $12$, which leads to a contradiction.
Hence, we conclude that every supported bundle of the $c_1$-reporter satisfies $c_0(Y)\leq 11$.
\end{proofof}

\subsection{Equal-Marginal Mechanisms}
\label{sec:equal-marginal-lb}

Finally, we show that equal-marginal mechanisms cannot guarantee a ratio better than $3/2$.
For any number $n\ge 2$ of agents, consider $n+1$ chores $e_1,\ldots,e_{n+1}$.
The costs are shown in Table~\ref{tab:lower_bound_equal_marginal}.
It is easy to verify that every agent has MMS equal to $1$.

\begin{table}[ht]
    \centering
    \begin{tabular*}{0.82\textwidth}{@{\extracolsep{\fill}}lccccc@{}}
        \toprule
        & $e_1$ & $e_2$ & $\cdots$ & $e_n$ & $e_{n+1}$ \\
        \midrule
        Agent $1$ & $1$ & $1/2$ & $\cdots$ & $1/2$ & $1/2$ \\
        Agent $2$ & $1/2$ & $1$ & $\cdots$ & $1$ & $1/2$ \\
        $\vdots$ & $\vdots$ & $\vdots$ & & $\vdots$ & $\vdots$ \\
        Agent $n$ & $1/2$ & $1$ & $\cdots$ & $1$ & $1/2$ \\
        \bottomrule
    \end{tabular*}
    \caption{Hard instance for equal-marginal mechanisms}
    \label{tab:lower_bound_equal_marginal}
\end{table}

A key property of the instance is that every allocation that is strictly better than $3/2$-MMS cannot allocate chore $e_1$ to agent $1$, because otherwise some agent among $\{2,3,\ldots,n\}$ must receive a total cost of at least $3/2$.
Consequently, no equal-marginal mechanism can guarantee an approximation ratio strictly smaller than $3/2$, because such a mechanism allocates $e_1$ to agent $1$ with probability $1/n>0$.

\section{Conclusion and Open Problems}
\label{sec:conclusion}
In this paper, we show that constant ex-post MMS guarantees can be achieved together with TIE for additive chores. 
The $k$-nomination mechanism gives ratio $2-1/\left(n-k\right)$ for $k\le\left\lfloor\left(n-1\right)/2\right\rfloor$, while a separate two-agent mechanism achieves the tight ratio $4/3$. For large $n$, we present an equal-marginal mechanism with a ratio of $1.97$. 
Our lower bounds show that the incentive requirement remains nontrivial: the ratio is at least $4/3$ for two agents and at least $13/12$ for every $n \geq 3$.

Several questions remain open. 
Most importantly, the gap between $1.97$ and the lower bounds is large for every $n\ge3$, and even the tight ratio for three agents is unknown.
By refining the grouping of agents for probability redistribution beyond the groups of size two and three used in this paper, and allowing the number of large chores to vary depending on the report, we can further improve the upper bound to $15/8 = 1.875$. However, the refined analysis is substantially more complex, and the resulting bound remains far from matching the lower bound. We therefore decided to present only the current, simpler analysis. 
Finally, it would be useful to understand whether richer nomination rules or new decompositions of truthful fractional allocations can yield a substantially smaller constant, or even approach the existing best-known MMS ratios without incentive constraints.

\section*{Declaration for the Use of AI}
We used OpenAI's GPT-6 Astra to assist with mathematical analysis, language editing, notation consistency checks, and LaTeX formatting. All mathematical statements, proofs, calculations, and citations were reviewed and verified by the authors, who take full responsibility for the content of the paper.

\newpage
\bibliographystyle{alpha}
\bibliography{ref}

\newpage
\appendix

\section{Proof of the Sorted-Slot Rounding Lemma}
\label{subsec:rounding}

\begin{proof}[Proof of \Cref{lem:sorted-slot-rounding}]
We first construct the sorted slots as described and build a weighted bipartite graph with chores on the left side and all agents' slots on the right side. 
For every chore and slot, we add an edge between them whenever a positive fraction of the chore lies in the slot, setting the edge weight equal to this fraction. By construction, the incident edge weights sum to one at each chore vertex and to at most one at each slot vertex. 
Let $k$ denote the total number of slots, which satisfies $|M| \le k \le |M| + n$. 
To balance the bipartite graph, we introduce $k - |M|$ dummy chore vertices, each assigned a total incident weight of one to cover the slot deficits (whose total sum is $k - |M|$). Consequently, the resulting edge weights constitute a feasible fractional perfect matching.

By the Birkhoff--von Neumann theorem~\cite{journals/tucuman/Birkhoff46,books/pup/vonNeumann53}, this fractional matching is a convex combination of integral perfect matchings. 
By sampling a matching with probability proportional to its convex coefficient and discarding dummy chores, we assign each chore to the owner of its matched slot. 
This guarantees that every chore is assigned exactly once and no slot receives more than one chore. 
Since the decomposition preserves edge weights, summing over agent $i$'s slots for chore $e$ yields $\Pr[e \in X_i] = x_{ie}$. 
Fix an agent $i \in N$. The contribution of agent $i$'s first slot is bounded by $\max_{e: x_{ie} > 0} \{c_i(e)\}$. 
For each subsequent slot, its cost is upper-bounded by the fractional cost of the preceding slot, because every chore fraction in the preceding slot is at least as expensive as any chore in the later slot, and each preceding slot has unit total fraction. 
Summing over all subsequent slots bounds their total contribution by $\sum_{e \in M} x_{ie}  \cdot c_i(e)$, which proves the general bound\footnote{Note that an agent with no positive fractions receives no chores, with the maximum understood to be zero.}.

If $\sum_{e \in H_{i,[n]}} x_{ie} = 1$, these fractions fill the first slot, which yields a cost of at most $c_{i,(1)}$.
All remaining slots contain only chores outside $H_{i,[n]}$ with costs at most $c_{i,(n+1)}$. Applying the same bounding argument to these remaining slots yields a total contribution of at most
\begin{equation*}
c_{i,(n+1)} + \sum_{e \in M \setminus H_{i,[n]}} x_{ie} \cdot c_i(e).
\end{equation*}
This proves the strengthened bound.
\end{proof}

\section{An Optimal \texorpdfstring{$4/3$}{4/3}-MMS Mechanism for Two Agents}
\label{sec:nomination-two}

In this section, we give the construction and analysis of the two-agent mechanism
outlined in Section~\ref{sec:nomination_overview}.
In this setting, for each agent $i\in N = [2]$, let $e_i$ be the unique chore in $H_{i,[1]}$ nominated by agent $i$.
If both agents nominate the same chore, the mechanism constructs a suitable partition of $M$ and assigns the two parts in a uniformly random order. 
Otherwise, each agent receives the chore nominated by her opponent, and the remaining chores are partitioned into two parts and assigned in a uniformly random order. 
To formalize this construction, we first establish the required partition bounds, starting with a balancing lemma for two cost functions.

\begin{lemma}\label{lem:two-dimensional-balancing}
    For two additive cost functions $c_1,c_2$ on a finite set $S$, there
    is a polynomial-time algorithm that computes a partition $S=P\cup Q$ such that
    \begin{equation*}
|c_i(P)-c_i(Q)|\le\max_{e\in S}\left\{c_i\!\left(e\right)\right\} \qquad\text{for both }i\in N.
\end{equation*}
\end{lemma}
\begin{proof}
    We proceed by induction on $|S|$. The base case $|S| = 0$ holds trivially.
    Suppose first that there exist two distinct chores $e, e' \in S$ that satisfy $c_1(e) \ge c_1(e')$ and $c_2(e) \ge c_2(e')$. 
    Then we can replace $e$ and $e'$ with a single virtual chore $e^*$ having costs
    $c_i(e^*) = c_i(e) - c_i(e')$ for $i \in N$.
    These costs are non-negative, and the maximum single-chore cost does not increase for either agent. By applying the induction hypothesis to this reduced set, we obtain a partition $P' \cup Q'$. Replacing $e^*$ with $e$ in its corresponding part and placing $e'$ in the opposite part preserves the cost difference $c_i(P) - c_i(Q)$ for both agents, so the desired bound holds for $S$.

    Otherwise, the two agents order every pair of chores in opposite directions. We can therefore index $S = \{e_1, \dots, e_{|S|}\}$ such that
    \begin{equation*}
c_1(e_1) \ge \cdots \ge c_1(e_{|S|}) \qquad \text{and} \qquad c_2(e_1) \le \cdots \le c_2(e_{|S|}).
\end{equation*}
    By assigning chores alternately to $P$ and $Q$ according to this order, we obtain
    \begin{equation*}
|c_1(P) - c_1(Q)| = c_1(e_1) - (c_1(e_2) - c_1(e_3)) - (c_1(e_4) - c_1(e_5)) - \cdots \le c_1(e_1).
\end{equation*}
    Applying the symmetric argument to agent $2$ by reading the ordered sequence in reverse yields $|c_2(P) - c_2(Q)| \le c_2(e_{|S|})$.

    The proof is constructive. Each recursive step reduces the number of chores by one, and a comparable pair $e,e'$ can be found by checking all pairs of remaining chores. Therefore, there are at most $|S|-1$ recursive steps, each requiring polynomial time. The resulting partition can thus be computed in time polynomial in the input bit size.
\end{proof}
In the following, we establish the partition bounds for both cases, beginning with coincident nominations. 
When $e_1=e_2=e$, we have $c_i\!\left(e\right)=c_{i,(1)}$ for both agents. 
We thus seek a partition of $M$ such that both parts cost each agent at most $4/3$ of her MMS.
Lemma~\ref{lem:common-top-balancing} establishes the existence of such a partition and guarantees that each part costs agent $i$ at most $4\mu_i/3$.
On the other hand, if $e_1 \neq e_2$, each agent receives the chore nominated by the other, while the remaining chores are partitioned into two parts and assigned uniformly at random. Lemma~\ref{lem:distinct-top-balancing} ensures that, regardless of which part agent $i$ receives, her total cost (including the assigned nomination) is at most $5\mu_i/4$.

\begin{lemma}\label{lem:common-top-balancing}
    If $e_1=e_2=e$, there is a polynomial-time algorithm that computes $P\subseteq M-e$ such that
    \begin{equation*}
\max\left\{c_i\!\left(P+e\right),c_i\!\left(M\setminus\left(P+e\right)\right)\right\}\le\frac43\cdot \mu_i \qquad\text{for both }i\in N.
\end{equation*}
\end{lemma}
\begin{proof}
    In this case, $e$ is the unique chore in $H_{i,[1]}$ for both agents and $c_i\!\left(e\right)=c_{i,(1)}$.
    We distinguish three cases according to how $c_i(e)$ compares with $c_i(M)/3$ for the two agents.

    Suppose first that $c_i(e)\le c_i(M)/3$ holds for both agents.
    We apply Lemma~\ref{lem:two-dimensional-balancing} to $M$.
    Let $A\cup B=M$ be the resulting partition and assume without loss of generality that $e\in A$.
    Since $e$ is a largest chore for each agent, we have
    \begin{equation*}
    \max\left\{c_i(A),c_i(B)\right\}
    \le \frac{c_i(M)+c_i(e)}{2}
    \le \frac{2}{3}\cdot c_i(M)
    \le \frac{4}{3}\cdot\mu_i.
    \end{equation*}
    Thus, $P=A-e$ is the desired subset.

    Suppose next that $c_i(e)\ge c_i(M)/3$ holds for both agents.
    Choose $P=\emptyset$.
    The two parts are then $\{e\}$ and $M-e$.
    For both agents, we have
    \begin{equation*}
    c_i(e)\le\mu_i
    \qquad\text{and}\qquad
    c_i(M-e)\le\frac{2}{3}\cdot c_i(M)\le\frac{4}{3}\cdot\mu_i.
    \end{equation*}

    It remains to consider the mixed case.
    By symmetry, assume $c_1(e)>\frac{c_1(M)}{3}$ and $c_2(e)<\frac{c_2(M)}{3}$.
    Define $\tau:=\frac{c_2(M)}{3}-c_2(e)>0$.
    We construct three pairwise disjoint subsets $S_1,S_2,S_3\subseteq M-e$.
    For $r=1,2,3$, greedily add arbitrary unused chores to $S_r$ until $c_2(S_r)\ge\tau$.
    Since $e$ is a largest chore for agent $2$, every chore in $M-e$ has cost at most $c_2(e)$.
    By the minimality of the greedy construction, we have
    \begin{equation*}
    \tau\le c_2(S_r)<\tau+c_2(e)=\frac{c_2(M)}{3}.
    \end{equation*}
    The three sets can indeed be constructed.
    After removing any two of them, the remaining $c_2$-cost is strictly larger than
    \begin{equation*}
    c_2(M-e)-\frac{2}{3}\cdot c_2(M)
    =\frac{c_2(M)}{3}-c_2(e)
    =\tau.
    \end{equation*}

    Among $S_1,S_2,S_3$, choose $P$ with minimum $c_1$-cost.
    Since these sets are pairwise disjoint, we have $c_1(P)\le\frac{c_1(M-e)}{3}$.
    For agent $1$, we therefore have
    \begin{equation*}
    c_1(P+e)
    \le c_1(e)+\frac{c_1(M-e)}{3}
    =\frac{c_1(M)+2\cdot c_1(e)}{3}
    \le\frac{4}{3}\cdot\mu_1,
    \end{equation*}
    where the last inequality follows from $c_1(M)\le2\cdot\mu_1$ and $c_1(e)\le\mu_1$.
    Moreover,
    \begin{equation*}
    c_1\!\left(M\setminus\left(P+e\right)\right)
    \le c_1(M-e)
    <\frac{2}{3}\cdot c_1(M)
    \le\frac{4}{3}\cdot\mu_1.
    \end{equation*}

    For agent $2$, the lower bound $c_2(P)\ge\tau$ gives
    \begin{equation*}
    c_2\!\left(M\setminus\left(P+e\right)\right)
    =c_2(M-e)-c_2(P)
    \le\frac{2}{3}\cdot c_2(M)
    \le\frac{4}{3}\cdot\mu_2.
    \end{equation*}
    The upper bound $c_2(P)<c_2(M)/3$ gives
    \begin{equation*}
    c_2(P+e)
    <c_2(e)+\frac{c_2(M)}{3}
    <\frac{2}{3}\cdot c_2(M)
    \le\frac{4}{3}\cdot\mu_2.
    \end{equation*}
    Hence $P$ satisfies the desired bounds for both agents.

    The construction uses only cost comparisons, one call to Lemma~\ref{lem:two-dimensional-balancing}, and at most three greedy scans of the chores. It therefore runs in time polynomial in the input bit size.
\end{proof}

Next, we consider the case where $e_1 \neq e_2$.

\begin{lemma}\label{lem:distinct-top-balancing}
    If $e_1\ne e_2$, there is a polynomial-time algorithm that computes a partition
    $M\setminus\{e_1,e_2\}=P\cup Q$ such that
    \begin{equation*}
    \max\left\{c_i\!\left(P+e_{3-i}\right),c_i\!\left(Q+e_{3-i}\right)\right\}\le\frac54\cdot \mu_i \qquad\text{for }i\in N.
\end{equation*}
\end{lemma}
\begin{proof}
    Let $R=M\setminus\{e_1,e_2\}$. If $R=\emptyset$, take
    $P=Q=\emptyset$.
    Otherwise, Lemma~\ref{lem:two-dimensional-balancing} gives a
    partition $R=P\cup Q$ that simultaneously satisfies, for both agents,
    \begin{equation*}
        \max\left\{c_i\!\left(P\right),c_i\!\left(Q\right)\right\} =\frac{c_i(R)+|c_i(P)-c_i(Q)|}{2} \le\frac{c_i(R)+\max_{e\in R}\left\{c_i\!\left(e\right)\right\}}{2}.
    \end{equation*}     
    We now bound the cost after adding the other agent's nominated
    chore $e_{3-i}$. Fix an agent $i$. Since $e_i$ is her
    largest chore, every $e\in R$ satisfies $c_i(e)\le c_i(e_i)$.
    We also have $c_i(e)+c_i(e_{3-i})\le\mu_i$ as at least two of $\{e_i,e_{3-i},e\}$ must belong to the same bundle under any two-partition of $M$.
    Combining the two bounds gives $\max_{e\in R}\left\{c_i\!\left(e\right)\right\} \le\min\left\{c_i\!\left(e_i\right),\mu_i-c_i\!\left(e_{3-i}\right)\right\}$.
    Moreover, $c_i(M)\le2\mu_i$ implies
    \begin{equation*}
        c_i(R)=c_i(M)-c_i(e_i)-c_i(e_{3-i}) \le2\mu_i-c_i(e_i)-c_i(e_{3-i}).
    \end{equation*}
    Substituting these bounds into the balanced-partition guarantee,
    we obtain
    \begin{equation*}
\begin{aligned}
        \max\left\{c_i\!\left(P+e_{3-i}\right),c_i\!\left(Q+e_{3-i}\right)\right\}
        &\le c_i(e_{3-i})+\frac{c_i(R)+\max_{e\in R}\left\{c_i\!\left(e\right)\right\}}2\\
        &\le\mu_i+\frac{c_i(e_{3-i})-c_i(e_i)
             +\min\left\{c_i\!\left(e_i\right),\mu_i-c_i\!\left(e_{3-i}\right)\right\}}2\\
        &=\mu_i+\frac12\cdot \min\left\{c_i\!\left(e_{3-i}\right),\mu_i-c_i\!\left(e_i\right)\right\}.
    \end{aligned}
\end{equation*}
    Finally, $c_i(e_{3-i})\le c_i(e_i)$, so the two quantities in
    the minimum sum to at most $\mu_i$. Their minimum is therefore
    at most $\mu_i/2$, which implies the desired bound
    $\mu_i+\frac12\cdot\frac{\mu_i}{2}=\frac54\mu_i$.

    The partition is obtained by a single call to Lemma~\ref{lem:two-dimensional-balancing}, and hence can be computed in time polynomial in the input bit size.
\end{proof}

\medskip

\begin{algorithm}[H]
    \caption{Two-agent truthful $4/3$-MMS mechanism}
    \label{alg:two-agent-nomination}
    \KwIn{Reported cost functions $(c_1,c_2)$}
    \KwOut{An allocation $\bX=\left(X_1,X_2\right)$}
    For each $i\in N$, let $e_i$ be the unique chore in $H_{i,[1]}$\;
    \eIf{$e_1=e_2$}{
        Let $e\gets e_1$ and compute $P\subseteq M-e$ by Lemma~\ref{lem:common-top-balancing}\;
        Set $A\gets P+e$ and $B\gets M\setminus A$\;
        Output $(A,B)$ or $(B,A)$, each with probability $1/2$\;
    }{
        Compute $P\cup Q=M\setminus\{e_1,e_2\}$ by Lemma~\ref{lem:distinct-top-balancing}\;
        Output $(P+e_2,Q+e_1)$ or $(Q+e_2,P+e_1)$, each with probability $1/2$\;
    }
\end{algorithm}

\medskip

\begin{theorem}\label{thm:two-agent-nomination}
    Mechanism~\ref{alg:two-agent-nomination} runs in polynomial time and is TIE, ex-ante EF,
    and ex-post $4/3$-MMS for two agents.
\end{theorem}
\begin{proof}
    The ex-post $4/3$-MMS guarantee follows from our analysis for the two cases depending on whether $e_1 = e_2$.
    To establish truthfulness, fix agent $i$'s true cost function $c_i$ and the other agent's reported nomination $e_{3-i}$. Under either execution branch, the marginal probability of agent $i$ receiving any specific chore $e \in M$ is given by
    \begin{equation*}
\Pr[e \in X_i] = \frac{1}{2} + \frac{1}{2} \cdot \mathbf{1}(e = e_{3-i}) - \frac{1}{2} \cdot \mathbf{1}(e = e_i).
\end{equation*}
    By linearity of expectation, agent $i$'s expected cost under either branch simplifies to
    \begin{equation*}
    \E[c_i(X_i)] = \frac{c_i(M)}{2} + \frac{c_i(e_{3-i})}{2} - \frac{c_i(e_i)}{2}.
    \end{equation*}
    Since agent $i$'s report influences her expected cost through the term $-c_i(e_i)/2$, reporting her true largest chore maximizes $c_i(e_i)$ and minimizes her expected total cost.
    Under truthful reporting, we have $\E[c_i(X_i)] \leq c_i(M)/2$, which implies ex-ante EF since $\E[c_i(X_{3-i})] \geq c_i(M)/2$.

    Finally, finding the two nominated chores and computing all required cost sums take polynomial time. Lemmas~\ref{lem:two-dimensional-balancing}, \ref{lem:common-top-balancing}, and~\ref{lem:distinct-top-balancing} provide polynomial-time algorithms for the required partitions. The remaining operations consist only of comparisons, set operations, and sampling a fair coin. Therefore, Mechanism~\ref{alg:two-agent-nomination} runs in time polynomial in the input bit size.
\end{proof}

\end{document}